\documentclass[11pt]{article}
\usepackage{lmodern}
\usepackage[T1]{fontenc}
\usepackage[utf8]{inputenc}
\usepackage{microtype}
\usepackage{geometry} 
\usepackage{graphicx}
\usepackage[english]{babel}
\usepackage{amsmath}
\usepackage{amssymb}
\usepackage{amsthm}
\usepackage{bbm}
\usepackage{color, colortbl}
\usepackage[font=small]{caption} 
\usepackage{siunitx}
\allowdisplaybreaks
\usepackage{natbib}
\usepackage{tikz}
\usepackage{comment}

\tikzstyle{empty node}=[draw,inner sep=2] 
\tikzstyle{solid node}=[circle,draw,inner sep=2,fill=black] 
\tikzstyle{hollow node}=[circle,draw,inner sep=2] 
\tikzstyle{hollow node1}=[rectangle,draw,inner sep=2] 
\tikzset{edge/.style = {->,> = latex}}
\theoremstyle{plain}
\newtheorem{thm}{Theorem}[section]
\newtheorem{prop}[thm]{Proposition}

\newtheorem{cor}[thm]{Corollary}

\graphicspath{ {pictures/} }

\renewcommand{\P}{\mathbb{P}}
\newcommand{\E}{\mathbb{E}}
\newcommand{\G}{\mathcal{G}}
\newcommand{\T}{\mathcal{T}}

\newcommand{\indep}{\perp\!\!\!\perp}
\newcommand{\argmin}{\operatornamewithlimits{\arg\min}}
\newcommand{\argmax}{\operatornamewithlimits{\arg\max}}
\newcommand{\dto}{\stackrel{d}{\longrightarrow}}
\newcommand{\path}[2]{({#1} \rightsquigarrow {#2})}
\definecolor{bll}{RGB}{209,229,255}
\providecommand{\keywords}[1]
{
	\small	
	\textbf{\textit{Keywords---}} #1
}
\newcommand{\js}[1]{\textcolor{magenta}{\sffamily\footnotesize [JS: {#1}]}}
\newcommand{\sa}[1]{\textcolor{blue}{\sffamily\footnotesize [SA: {#1}]}}

\title{Inference on extremal dependence in the domain of attraction of a structured Hüsler--Reiss distribution motivated by a Markov tree with latent variables}
\author{Stefka Asenova\thanks{Corresponding author. UCLouvain, LIDAM/ISBA, Voie du Roman Pays 20, 1348 Louvain-la-Neuve, Belgium. E-mail: stefka.asenova@uclouvain.be} \and Gildas Mazo\thanks{MaIAGE, INRA, Université Paris-Saclay 78350, Jouy-en-Josas, France. E-mail: gildas.mazo@inra.fr} \and Johan Segers\thanks{UCLouvain, LIDAM/ISBA, Voie du Roman Pays 20, 1348 Louvain-la-Neuve, Belgium. E-mail: johan.segers@uclouvain.be}}
\date{\today}

\usepackage{hyperref}

\begin{document}

\maketitle

\begin{abstract}
	A Markov tree is a probabilistic graphical model for a random vector indexed by the nodes of an undirected tree encoding conditional independence relations between variables.
	One possible limit distribution of partial maxima of samples from such a Markov tree is a max-stable Hüsler--Reiss distribution whose parameter matrix inherits its structure from the tree, each edge contributing one free dependence parameter.
	Our central assumption is that, upon marginal standardization, the data-generating distribution is in the max-domain of attraction of the said Hüsler--Reiss distribution, an assumption much weaker than the one that data are generated according to a graphical model.
	Even if some of the variables are unobservable (latent), we show that the underlying model parameters are still identifiable if and only if every node corresponding to a latent variable has degree at least three.
	Three estimation procedures, based on the method of moments, maximum composite likelihood, and pairwise extremal coefficients, are proposed for usage on multivariate peaks over thresholds data when some variables are latent.
	A typical application is a river network in the form of a tree where, on some locations, no data are available. We illustrate the model and the identifiability criterion on a data set of high water levels on the Seine, France, with two latent variables. The structured Hüsler--Reiss distribution is found to fit the observed extremal dependence patterns well. The parameters being identifiable we are able to quantify tail dependence between locations for which there are no data. 
\end{abstract}

\keywords{multivariate extremes; tail dependence; graphical models; latent variables; Hüsler--Reiss distribution; Markov tree; tail tree; river network}

\section{Introduction}

A major topic in multivariate extreme value theory is the modeling of tail dependence between a finite number of variables. Informally, tail dependence represents the degree of association between the extreme values of these variables. Probabilistic graphical models \citep{laurit, koller_fried, martin_wain}, are distributions which embody a set of conditional independence relations and have a graph-based representation, according to which the nodes of the graph are associated to the variables and the set of edges encode the conditional independence relations. The intersection of the two fields, extreme value theory and probabilistic graphical models, gives rise to the study of the tail behavior of graphical models.

Consider a river network where the interest is in extreme water levels or water flow in relation to flood risks. Figure~\ref{fig:seine} illustrates part of the Seine network. The graph fixed by the seven labeled nodes and the river channels between them can be a base for building a model for extremal dependence between the water levels at these sites. 

Hydrological data are often used to fit models for multivariate extremes based on graphs. Water flows of the Bavarian Danube are analyzed in \citet{engelke+h:2020}. \citet{joe} study water flows of the Fraser river, British Colombia. Precipitation data in the Japanese archipelago is treated in \citet{justin}, where the model is based on a spatial grid viewed as an ensemble of trees. Other extreme-value models involving graphs appear in \citet{eks16} and \citet{joe}, who study financial data under different models. The first paper uses max-linear models on a directed acyclic graph (DAG) \citep{gissibl}, and the second one a 1-factor model. \citet{klup_sonmez} introduce an infinite max-linear model to analyze the distribution of extreme opinions in a social network.

Relatively recently the relation between extreme value distributions and conditional independence assumptions has been given theoretical relevance. The earliest is the article of \citet{gissibl} introducing max-linear models as structural equation models on a DAG, followed by the regularly varying Markov trees in \citet{mazo} and the extremal graphical models in \citet{engelke+h:2020} based on multivariate Pareto distributions. Earlier, \citet{pap161} showed that for a max-stable random vector with positive and continuous density, conditional independence implies unconditional independence, thereby concluding that a broad class of max-stable distributions does not exhibit an interesting Markov structure. 


A key object of our paper is the multivariate Hüsler--Reiss distribution \citep{hr} with parameter matrix having a particular structure linked to a tree as specified in Eq.~\eqref{eqn:lambda}. The structure is motivated by the fact that the max-domain of attraction of the said Hüsler--Reiss distribution contains certain regularly varying Markov trees. The latter property follows from results in \citet{mazo} and sets our work apart from the extremal graphical models in \citet{engelke+h:2020}, who impose a non-standard conditional independence relation on the multivariate Pareto distribution associated to a max-stable distribution, but without regard for the latter's max-domain of attraction. Still, it turns out that for trees, the structured Hüsler--Reiss models in \citet{engelke+h:2020} and in our paper are the same, as explained in Appendix~\ref{app:eh}. Another structured Hüsler--Reiss distribution based on trees is proposed in \citet{joe}. The form they propose is genuinely different from ours, however, as explained in detail in Appendix~\ref{app:LeeJoe}. 

We consider random samples from the distribution of a random vector $\xi = (\xi_v, v \in V)$ with continuous margins whose variables are indexed by the node set $V = \{1,\ldots,d\}$ of an undirected tree with edge set~$E$. After marginal standardization to the unit-Pareto distribution, we assume that the random vector is in the max-domain of attraction of the tree-structured Hüsler--Reiss distribution described in the previous paragraph. We emphasize that we do not assume that $\xi$ itself satisfies any conditional independence relations with respect to the tree.
The tree only comes into play via the imposed structure on the parameter matrix of the max-stable Hüsler--Reiss distribution containing the distribution of the standardized version of $\xi$ in its max-domain of attraction.

The main result and contribution of our paper is a criterion for identifiability of all $d-1$ parameters $\theta_e \in (0, \infty)$ for $e \in E$ of the tree-structured $d$-variate Hüsler--Reiss distribution in case some of the $d$ variables are latent (unobservable). 
To illustrate why the problem of latent variables is relevant, consider again the Seine network on Figure~\ref{fig:seine}. The red dots designate junctions of two river channels (conversely, in a river delta, a channel could split into several ones). No measurement stations being present there, we cannot observe the water levels at those locations. We propose to treat those water levels as latent variables. The question is then whether it is still possible to identify all $d-1$ parameters. The answer is a surprisingly simple identifiability criterion: it is necessary and sufficient that all nodes indexing latent variables have degree at least three. The important practical implication is that, provided the criterion is met, the latent variables can be included in the model, reflecting the dependence structure more accurately than when they would have been ignored.



\begin{figure}
    \centering
    \includegraphics[height=.4\textheight]{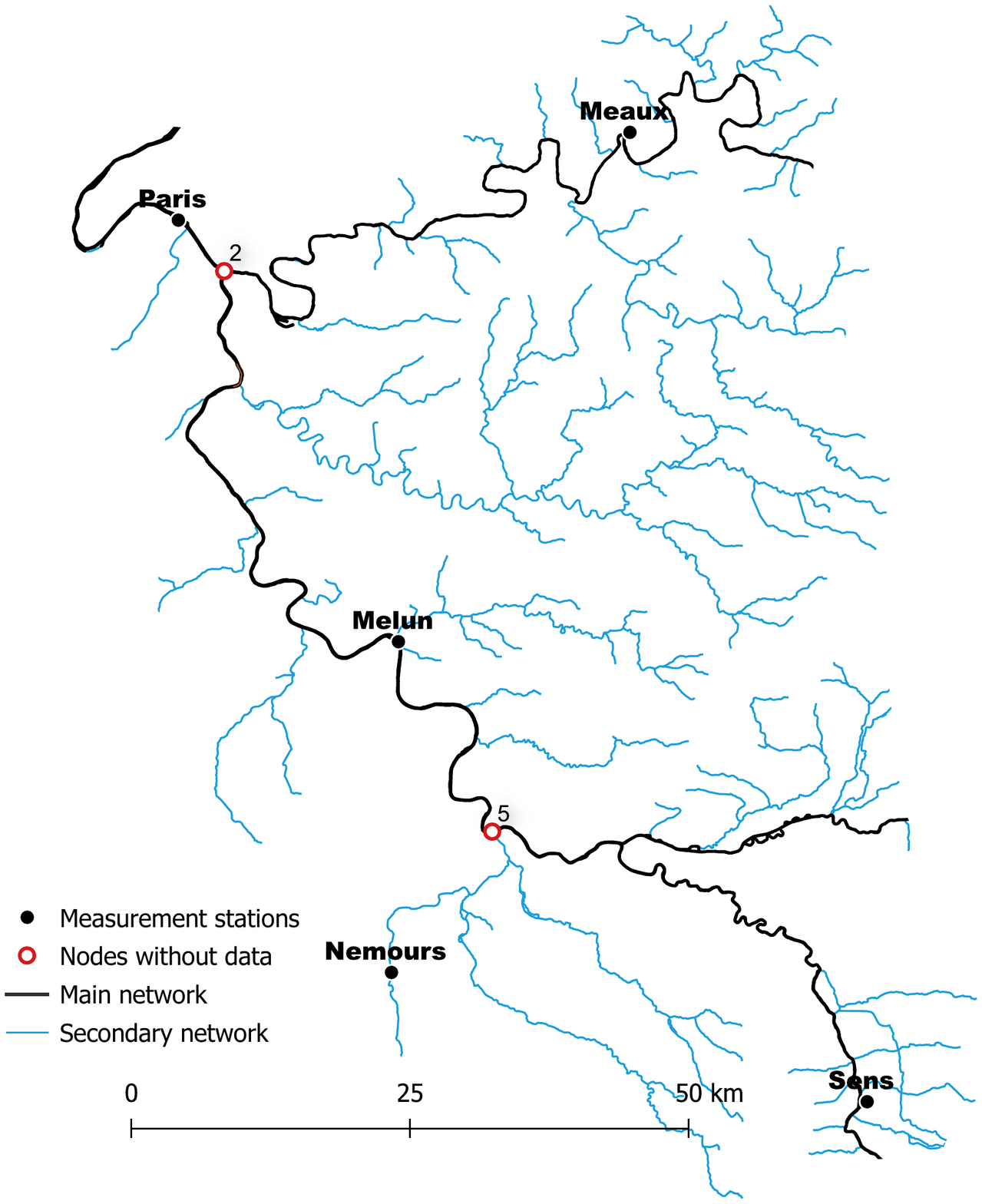}
    \caption{Seine network. The data is from the web-site of Copernicus Land Monitoring Service: \href{https://land.copernicus.eu/imagery-in-situ}{https://land.copernicus.eu/imagery-in-situ}.}
    \label{fig:seine}
\end{figure}

Given a random sample from a distribution in the max-domain of attraction of the tree-structured Hüsler--Reiss distribution, we propose
three types of estimators of the edge parameters: a first one called method of moments estimator (MME) is based on the estimator proposed in \citet{engelke}, a second one is based on the composite likelihood function (composite likelihood estimator or CLE) and the third one is essentially the pairwise extremal coefficient estimator (ECE) introduced in \citet{eks16}. 
All estimators proposed allow for the fact that some of the $d$ variables are latent, provided the identifiability criterion is met.

We illustrate the method by a detailed analysis of data on high water levels at several locations of the Seine network. The network is represented schematically as a tree with seven nodes indexing five observable variables and two latent ones. As the identifiability criterion is met, we can estimate the six dependence parameters of the tree-structured Hüsler--Reiss distribution, each parameter corresponding to an edge in the tree. For the three proposed estimators we compute parameter estimates and confidence intervals. We assess the goodness-of-fit by comparing the model output with various non-parametric measures of tail dependence. Finally, we compare the fitted tail dependence model incorporating latent variables with a model where the latent variables are ignored. 
The outline of the paper is as follows: Section~\ref{sec:2} presents some general theory and describes the model to which the identifiability criterion is applied. The latter is the focus of Section~\ref{sec:latent}. Section~\ref{section:estim} introduces the three estimators, used for statistical inference and Section~\ref{sec:Seine} is dedicated to the study of high water levels on the Seine network. Concluding remarks and perspectives for further research are discussed in Section~\ref{sec:concl}. The Appendix provides proofs that are not in the text, a numerical comparison between our structured Hüsler--Reiss method and the one of \citet{joe}, clarification on the relationship between the different objects in our paper and the objects in \citet{engelke+h:2020}, some simulation results which aim at comparing the different estimators, and details about some estimation procedures and the data preprocessing.

\section{The model -- definition and properties} \label{sec:2}

\subsection{Preliminaries} 
\label{ssec:prelim}


\paragraph{Multivariate extremes.}
Let $V = \{1, \ldots, d\}$ for some integer $d \ge 2$.
A $d$-variate max-stable distribution $G$ is called simple if its margins are unit-Fréchet, that is, a random vector $Z$ with distribution $G$ satisfies $\P(Z_v \leq x) = \exp(-1/x)$ for $x \in (0, \infty)$ and $v \in V$.
Let $X = (X_v, v \in V)$ be a random vector with unit-Pareto margins, i.e., $\P(X_v\leq x)=1-1/x$ for $x\in[1,\infty)$ and $v \in V$. Let $X_i = (X_{v,i}, v \in V)$ for $i = 1, \ldots, n$ be an independent random sample from the distribution of $X$. We say that $X$ belongs to the max-domain of attraction of the simple max-stable distribution $G$, notation $X \in D(G)$, if
\begin{equation*} 
\lim_{n\to\infty}
\P\left(\max_{i=1,\ldots, n}X_{v,i}\leq n z_v, v\in V\right)
=
G(z), \qquad z \in (0, \infty)^d.
\end{equation*}
For more background on max-stable distributions and their domains of attractions, we refer to the reader to \citet[Chapter~5]{resnick} and \citet[Chapter~6]{de2007extreme}.


Throughout the paper the stable tail dependence function (stdf) $l$ of $G$ or $X \in D(G)$ will appear frequently. It is defined as 
\begin{equation}
\label{eq:stdf}
	l(x)
	= \lim_{t \to \infty} t \, \bigl(1 - \P( X_v \le t / x_v, v \in V )\bigr) 
	= -\ln G(1/x_v, v\in V), 
	\qquad x \in [0,\infty)^d,
\end{equation}
with the obvious limit interpretation if $x_v = 0$ for some $v \in V$.
The stdf is closely linked to the exponent function of a simple max-stable distribution in \citet[Eq.~(2.4)]{coles+t:1991}. It is introduced and studied in \citet{huang1992statistics} and \citet{drees1998}; see also later literature in \citet[Chapter~6]{de2007extreme} and \citet[Chapter~8]{beirlant2004statistics}. The stdf evaluated at $x_J=(\mathbbm{1}_{\{j\in J\}}, j\in V)$ is known as an extremal coefficient, of which we make use in Sections~\ref{section:estim} and~\ref{sec:Seine}.

One of the main objects in our paper is the multivariate Hüsler--Reiss distribution. This absolutely continuous max-stable distribution was introduced in \citet{hr} and remains a popular parametric model in recent literature \citep{genton+m+s:2011, huser+d:2013, asadi, engelke, eks16, joe}. It arises as the limiting distribution of partial maxima of a triangular array of row-wise independent and identically distributed random vectors from a multivariate normal distribution with correlation matrix $\rho(n)$ depending on the sample size $n$. In particular, assume that
\[
  \lim_{n\rightarrow\infty}\big(1-\rho_{ij}(n)\big)\ln n
  =
  \lambda^2_{ij} \in (0, \infty)
\]
for every pair of variables $i, j \in V$ and let $\Lambda=(\lambda^2_{ij})_{i,j\in V}$ denote this limiting matrix. Note that $\lambda_{ii}^2 = 0$ for every $i \in V$. For every subset $W\subseteq V$ and any element $u\in W$ let $\Gamma_{W,u}(\Lambda)$ be the square matrix of size $|W|-1$ with elements
\begin{equation} \label{eq:hrdist}
   \big(\Gamma_{W,u}(\Lambda)\big)_{ij}
    =
    2(\lambda_{iu}^2
    +
    \lambda_{ju}^2
    -
    \lambda^2_{ij}),
    \qquad i,j\in W\setminus u.
\end{equation}
\citet{nikoloulopoulos+j+l:2009} and later \citet{genton+m+s:2011} and \citet{huser+d:2013} show that the cumulative distribution function (cdf) as deduced by \citet{hr} can be written as 
\begin{equation} \label{eq:mvhr_short}
    H_{\Lambda}(z)
    =
    \exp\left\{-
    \sum_{u\in V}
    \frac{1}{z_u}\Phi_{d-1}\left(
     \ln\frac{z_v}{z_u} +2\lambda^2_{uv}, v\in V\setminus u; \Gamma_{V,u}(\Lambda)
      \right)
      \right\},
    \qquad z \in (0, \infty)^d, 
\end{equation}
where $\Phi_p(\,\cdot\,; \Sigma)$ denotes the $p$-variate zero mean Gaussian cdf with covariance matrix $\Sigma$. The distribution $H_\Lambda$ in~\eqref{eq:mvhr_short} is a simple max-stable distribution. In particular, its margins are unit-Fréchet, whereas \citet{hr} originally proposed the distribution in terms of Gumbel margins.

Multivariate margins of the $d$-variate Hüsler--Reiss distribution are Hüsler--Reiss distributions too. The corresponding parameter matrix is obtained by selecting the appropriate rows and columns in the original parameter matrix \citep[see, e.g.,][Example~7]{engelke+h:2020}. In particular, if $X \in D(H_\Lambda)$ and if $U \subseteq V$ is non-empty, the stdf $l_U$ of $X_U = (X_u, u \in U)$ is
\begin{equation} 
\label{eq:stdf_hr}
l_U(x)
=
\sum_{u\in U}
{x_u}\, \Phi_{|U \setminus u|}\left(
\ln\frac{x_u}{x_v} +2\lambda^2_{uv}, v \in U \setminus u; \Gamma_{U,u}(\Lambda)
\right),
\qquad x \in [0, \infty)^U.
\end{equation}
Here we write $U \setminus u$ instead of $U \setminus \{u\}$. In case $x_u = 0$ for some $u \in U$, the corresponding term in the sum in \eqref{eq:stdf_hr} vanishes.

\paragraph{Trees.}
We will need some notions from graph theory. A graph is a pair $\mathcal{G} = (V, E)$ where $V = \{1, \ldots, d\}$ is the set of nodes or vertices and $E \subseteq \{(a, b) \in V \times V : a \ne b \}$ is the set of edges. Edges will also be denoted by $e = (a, b) \in E$. The number of vertices in a subset $U\subseteq V$ will be denoted by $|U|$, while $d$ is reserved for $|V|$ only. A graph is undirected if $(a, b) \in E$ is equivalent to $(b, a) \in E$. A path $\path{u}{v}$ from node $u$ to node $v$ is a collection $\{(u_0, u_1), (u_1, u_2), \ldots, (u_{n-1}, u_n)\}$ of distinct, directed edges such that $u_0 = u$ and $u_n = v$. An undirected tree is an acyclic undirected graph $\T = (V, E)$ such that for every pair of distinct nodes $a$ and $b$ there is a unique path $\path{a}{b}$.


\subsection{Model definition } \label{ssec:introX}

Let $\T = (V, E)$ be an undirected tree with node set $V = \{1, \ldots, d\}$ and let $\xi=(\xi_v, v\in V)$ be a random vector with joint cdf $F$ and continuous margins $F_v(z) = \P(\xi_v \le z)$ for $z \in \mathbb{R}$ and $v \in V$.
Let the random vector $X=(X_v, v\in V)$ be defined as $X_v=1/\big(1-F_v(\xi_v)\big)$ for every $v\in V$. 
Because the functions $F_v$ for $v\in V$ are continuous, the marginal distributions of $X$ are unit-Pareto.   


We assume that $X$ is in the max-domain of attraction of the Hüsler--Reiss distribution $H_\Lambda$ in~\eqref{eq:mvhr_short} with $\Lambda = (\lambda_{ij}^2)_{i,j \in V}$ having the following structure linked to the tree $\T$: there exists a vector $\theta = (\theta_e)_{e \in E}$ of positive scalars with $\theta_{ab} = \theta_{ba}$ and such that $\Lambda = \Lambda(\theta)$ where
%
\begin{equation} 
\label{eqn:lambda}
	\big(\Lambda(\theta)\big)_{ij}
	= \lambda^2_{ij}(\theta)
	= \frac{1}{4}\sum_{e \in \path{i}{j}} \theta_e^2\, , \qquad i,j\in V, \ i \ne j.
\end{equation}
The assumption can thus be written compactly as $X\in D(H_{\Lambda(\theta)})$ for some $\theta \in (0, \infty)^E$.

The motivation for the proposed structure is that $H_{\Lambda(\theta)}$ contains in its max-domain of attraction a certain graphical model with respect to $\T$ as explained in Section~\ref{ssec:Y}. Still, it is to be noted that, despite the structure of the parameter matrix, $H_{\Lambda(\theta)}$ itself does not and cannot satisfy any Markov properties with respect to the tree $\T$: by \citet{pap161}, max-stable distributions with continuous joint densities cannot possess any non-trivial conditional independence properties.

In the parametrization in \eqref{eqn:lambda} the extremal dependence in $\xi$ and in $X$ depends on a vector $\theta=(\theta_e, e\in E)$ of $d-1$ free parameters, indexed by the edges of the tree. 
The main theme in this paper concerns inference on the parameter vector $\theta$ in case some of the variables $\xi_v$ are latent (unobservable). The first question is whether all edge parameters $\theta_e$ are still identifiable from \eqref{eq:stdf_hr} when $\Lambda = \Lambda(\theta)$ and when $U \subsetneq V$ contains the indices of variables that can still be observed. For the Seine network in Figure~\ref{fig:seine}, for instance, there are $d = 7$ variables in total, of which two are latent. A necessary and sufficient criterion for parameter identifiability is given in Proposition~\ref{prop:identif} below. Provided the criterion is fulfilled, the second question is how to estimate the parameters. Three estimation methods are proposed in Section~\ref{section:estim} and illustrated in Section~\ref{sec:Seine}.

Note that the random vector $\xi$ itself does not necessarily belong to the max-domain of attraction of some max-stable distribution. The reason is that we do not impose that the marginal distributions of $\xi$ are in the max-domain of attraction of some univariate extreme value distributions. To focus on the tail dependence of $\xi$, we standardize its margins and formulate the assumption in terms of $X$.


\subsection{Motivation of the structured Hüsler--Reiss model}
\label{ssec:Y}

To motivate the structured Hüsler--Reiss parameter matrix $\Lambda(\theta)$ in \eqref{eqn:lambda}, we construct a graphical model $Z^*$ that satisfies the global Markov property with respect to the undirected tree $\T = (V, E)$ and such that $Z^* \in D(H_{\Lambda(\theta)})$. Besides serving as a motivation, the auxiliary model $Z^*$ plays another important role: in view of \citet[Theorem~2]{mazo} we are able to project certain asymptotic properties that hold for $Z^*$ to $X$.  

For disjoint subsets $A, B, C$ of $V$, the expression $A\indep_{\T} B\mid C$ means that $C$ separates $A$ from $B$ in $\T$, also called graphical separation, i.e., all paths from $A$ to $B$ pass through at least one vertex in $C$. Let $Z^*$ be defined on a probability space $(\Omega, \mathcal{B}, \P)$. Conditional independence of $Z^*_A$ and $Z^*_B$ given $Z^*_C$ will be denoted by $Z^*_A \indep_{\P} Z^*_B \mid Z^*_C$; here $Z^*_A = (Z^*_a, a \in A)$ and so on. If $P=\P(Z^* \in \,\cdot\,)$ is the law of $Z^*$, we say that the tree $\T$ is an independence map (I-map) of $P$ if for any disjoint subsets $A,B,C$ of $V$ it holds that
\begin{equation} \label{eq:gmp}
A\indep_{\T}B\mid C
\implies 
Z^*_A\indep_{\P}Z^*_B\mid Z^*_C
\end{equation}
\citep{koller_fried}. 
This assumption is equivalent to the assumption that $Z^*$ obeys the global Markov property with respect to $\T$ \citep{laurit}. 
 


The law of the random vector $Z^*=(Z^*_v, v\in V)$ is defined by the following two assumptions: 
\begin{itemize}
	\item[(Z1)] $Z^*$ satisfies the global Markov property~\eqref{eq:gmp} with respect to the undirected tree $\T = (V, E)$;
	\item[(Z2)] every pair of variables $(Z^*_a, Z^*_b)$ on adjacent nodes $(a,b)=e\in E$ has a bivariate Hüsler--Reiss distribution with parameter $\theta_e \in (0, \infty)$ and unit-Fréchet margins, i.e., the special case of \eqref{eq:stdf_hr} with $U = \{a, b\}$ and $\lambda_{ab}^2 = \theta^2_e/4$.
\end{itemize}
The law of $Z^*$ is absolutely continuous and its joint density function factorizes in terms of the bivariate Hüsler--Reiss densities along pairs of variables on adjacent nodes through the Hammersley--Clifford theorem; see Appendix~\ref{app:simu} where we describe how to sample from $Z^*$. Moreover, for $e = (a, b) \in E$ and if $Z$ has distribution $H_{\Lambda(\theta)}$, the law of $(Z^*_a, Z^*_b)$ is the same as the one of $(Z_a, Z_b)$. However, unless $d = 2$, the law of $Z^*$ is itself not max-stable and thus not equal to the one of $Z$. One way to see this is to note that by \citet{pap161}, the law of $Z$ cannot satisfy the global Markov property with respect to $\T$.

Let $(M_e, e \in E)$ be a random vector of independent lognormal random variables with $\ln M_e \sim \mathcal{N}(-\theta^2_e/2, \theta_e^2)$ for each $e \in E$. In view Theorem~1 and Corollary~1 in \citet{mazo}, we have the convergence in distribution
\begin{equation} \label{eqn:YtoXi}
(Z^*_v/Z^*_u, v\in V\setminus u)\mid Z^*_u>x
\dto
(\Xi_{u,v}, v\in V \setminus u)
=\left({\textstyle\prod_{e \in \path{u}{v}}} M_{e}, \, v\in V\setminus u\right),
\qquad x \to \infty,
\end{equation}
for every $u \in V$.
For every $u\in V$ the vector $(\Xi_{u,v}, v\in V\setminus u)$ is called a tail tree. The multiplicative structure in \eqref{eqn:YtoXi} goes back to the theory of extremes of Markov chains due to \citet{smith_1992}, \citet{perfekt}, \citet{yun} and \citet{seg07}. Note that a chain can be seen as a tree with a single branch.

The vector $(\ln \Xi_{u,v},v\in V\setminus u)$ is a linear transformation of a Gaussian random vector and is therefore itself Gaussian. Its mean vector $\mu_{V,u}(\theta)$ and its covariance matrix $\Sigma_{V,u}(\theta)$ have elements
\begin{align}
\label{eqn:muY}
	\{\mu_{V,u}(\theta)\}_v 
	&=-\frac{1}{2}\sum_{e \in \path{u}{v}} \theta_{e}^2\,, &v\in V \setminus u, \\
\label{eqn:sigmaY}
	\{\Sigma_{V,u}(\theta)\}_{ij}
	&=\sum_{e \in \path{u}{i} \cap \path{u}{j}} \theta_{e}^2\,,
	&i,j\in V\setminus u\, .
\end{align} 
Hence for every $u\in V$ and as $x\rightarrow\infty$, we have the convergence in distribution
\begin{equation} \label{eqn:limitY}
(\ln Z^*_v-\ln Z^*_u, v\in V\setminus u)\mid Z^*_u>x \dto
(\ln \Xi_{u,v}, v\in V \setminus u)\sim \mathcal{N}_{|V \setminus u|}\big(\mu_{V,u}(\theta), \Sigma_{V,u}(\theta)\big), 
\end{equation}
where $\mathcal{N}_p$ is the $p$-variate normal distribution. By construction, $\Sigma_{V,u}(\theta)$ is a covariance matrix and hence positive semi-definite for any $\theta\in (0,\infty)^{d-1}$; it is actually positive definite since the vector $(\ln \Xi_{u,v}, v \in V \setminus u)$ is the result of an invertible linear transformation applied to the vector $(\ln M_e, e \in E)$ of independent and non-degenerate normal random variables. 
The matrix $\Sigma_{V,u}(\theta)$ is moreover the same as the matrix $\Gamma_{W,u}(\Lambda)$ in \eqref{eq:hrdist} with $W = V$ and $\Lambda = \Lambda(\theta)$ in \eqref{eqn:lambda}: 
%
%
\begin{align}
\nonumber
\big\{\Sigma_{V,u}(\theta)\big\}_{ij}
&=
\sum_{e \in \path{u}{i} \cap \path{u}{j}}\theta_e^2
=
\frac{1}{2}\left(\sum_{e \in \path{u}{i}}\theta_e^2
+
\sum_{e \in \path{u}{j}} \theta_e^2
-
\sum_{e \in \path{i}{j}}\theta_e^2\right)
\\&=
2(\lambda_{iu}^2+\lambda_{ju}^2-\lambda_{ij}^2)
=\big\{\Gamma_{V,u}\big(\Lambda(\theta)\big)\big\}_{ij}\,,
\qquad i,j\in V\setminus u.
\label{eq:Sigma2Gamma}
\end{align}
In the second equality it is needed to divide by two because the parameters on shared edges are added twice. 
In addition, the Hüsler--Reiss parameters $\lambda_{uv}^2$ are proportional to the means: 
\begin{equation}
\label{eq:mu2lambda}
    2\lambda_{uv}^2
    =\frac{1}{2}\sum_{e \in \path{u}{v}} \theta_e^2
    = -\{\mu_{V,u}(\theta)\}_v \,, \qquad v\in V\setminus u.
\end{equation}
\begin{prop} \label{prop:Ydomain}
	Let $\T = (V, E)$ be a tree.
	If the law of $Z^* = (Z^*_v, v \in V)$ is given by (Z1)--(Z2) above, then $Z^* \in D(H_{\Lambda(\theta)})$ with $\Lambda(\theta)$ in \eqref{eqn:lambda}.
\end{prop} 
The proof is given in Appendix~\ref{app:proofprop} and relies on the properties of $Z^*$ mentioned above, in particular on~\eqref{eqn:limitY}.
By constructing a graphical model with respect to $\T$ in the max-domain of attraction of $H_{\Lambda(\Theta)}$, we have argued that the latter is a sensible dependence model for extremes of graphical models on trees.
Moreover, it follows that any random vector $X = (X_v, v \in V)$ with unit-Pareto margins and in the max-domain of attraction of $H_{\Lambda(\theta)}$ shares property~\eqref{eqn:limitY} with $Z^*$.

\begin{cor}
	\label{cor:lnXdiffN}
	Let $\T = (V, E)$ be a tree and let $X=(X_v, v\in V)$ have unit-Pareto margins and belong to $D(H_{\Lambda(\theta)})$ with $\Lambda(\theta)$ as in \eqref{eqn:lambda} for a vector $\theta = (\theta_e, e \in E)$ of positive scalars.
	 Then for every $u\in V$, we have
	\begin{equation} \label{eq:XlimN}
	(\ln X_v-\ln X_u,{v\in V\setminus u}) \mid X_u>t
	\dto
	\mathcal{N}_{|V \setminus u|}\bigl(
	\mu_{V,u}(\theta), 
	\Sigma_{V,u}(\theta)
	\bigr), \qquad t \to \infty. 
	\end{equation}
\end{cor}

\begin{proof}
	The max-domain of attraction condition $X \in D(H_{\Lambda(\theta)})$ is known to be equivalent to convergence of the measures $t \, \P(X/t \in \,\cdot\,)$ as $t \to \infty$ to the exponent measure of $H_{\Lambda(\theta)}$ \citep[Proposition~5.17]{resnick}. Such measure convergence is in turn equivalent to convergence in distribution of $X / X_u \mid X_u > t$ as $t \to \infty$ for every $u \in V$ to a limit that can be written in terms of the said exponent measure \citep[Theorem~2]{mazo}. But for $X$ replaced by $Z^*$, the limiting conditional distribution was found to be a certain multivariate lognormal distribution in \eqref{eqn:YtoXi}. The equivalence between \eqref{eqn:YtoXi} and \eqref{eqn:limitY} with $Z^*$ replaced by $X$ is clear by the continuous mapping theorem.
\end{proof}

The random vector $X$ in Corollary~\ref{cor:lnXdiffN} does not need to be a  graphical model with respect to $\T$. The convergence in~\eqref{eq:XlimN} appears in \citet[Theorem~2]{engelke} for a general random vector with standardized margins and in the max-domain of attraction of a Hüsler--Reiss distribution. With Corollary~\ref{cor:lnXdiffN} we arrive at the same result but through the properties of the auxiliary model $Z^*$. The convergence in~\eqref{eq:XlimN} is used to build two estimators in the next section.

In \citet{engelke+h:2020}, a notion of conditional independence different from the classical one is introduced in the context of multivariate Pareto distributions. When specialized to the Pareto distribution associated to a max-stable Hüsler--Reiss distribution, it yields certain restrictions on the Hüsler--Reiss parameter matrix $\Lambda$. In case the conditional independence relations are the ones induced by a tree through graphical separation, the structure of the parameter matrix is the same as the one in \eqref{eqn:lambda}. We explain the connection in Appendix~\ref{app:eh}. Here we just emphasize that in \citet{engelke+h:2020}, no graphical model in the classical sense of the term is constructed that belongs to the max-domain of attraction of $H_{\Lambda(\theta)}$. The way we arrive at the structure of $\Lambda(\theta)$ via the graphical model $Z^*$ in (Z1)--(Z2) is thus entirely different from their approach.

Finally, quite another tree-induced structure of the Hüsler--Reiss parameter matrix is proposed in \citet{joe}. We provide a comparison in Appendix~\ref{app:LeeJoe}.

\section{Latent variables and parameter identifiability}
\label{sec:latent}

A typical application of our model arises in relation to quantities measured on river networks that have a tree-like structure. 
It is natural to associate a node to an existing measurement station or to locations where two river channels meet (junction) or one channel splits (split) even if there is no measurement station there. Stations are supposed to generate data for the quantity of interest, so for any node associated to a station there is a corresponding variable. In practice, junctions/splits may lack measurements, and this means that there are nodes in the tree with latent variables. Nodes with latent variables are those labelled~2 and~5 in the Seine network in Figure~\ref{fig:seine}. 

A naive approach to the presence of latent variables would be to ignore them, that is, to remove the corresponding nodes and all edges incident to them. This will yield a disconnected graph, making it necessary to add edges in some arbitrary way so as to obtain a tree again.
In Figure~\ref{fig:ignore_node} for instance, if node $2$ is suppressed, there are three possible ways to reconnect the remaining nodes and form a tree. Each implies a different structured Hüsler--Reiss parameter matrix and thus a different dependence model.

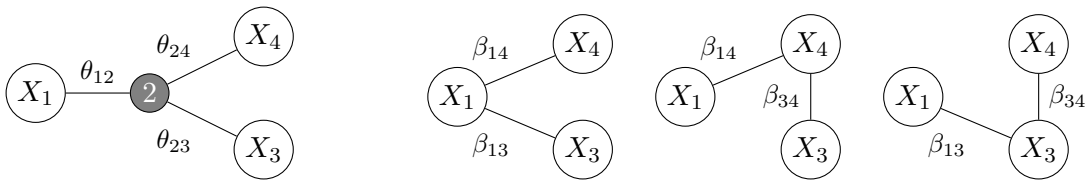
\begin{figure}[h]
	\centering
	\begin{tikzpicture}
	\node[hollow node](1){$X_1$}
	child[grow=right]{node [hollow node, fill=gray] (a) {\textcolor{white}{$2$}}
		child {node [hollow node] (a1) {$X_3$}}
		child {node [hollow node] (a2) {$X_4$}}
	};
	\path (1) -- (a) node [midway,auto=left] {\small $\theta_{12}$};
	\path (a) -- (a1) node [midway,auto=right] {\small $\theta_{23}$};
	\path (a) -- (a2) node [midway,auto=left] {\small $\theta_{24}$};
	\end{tikzpicture}
	\hspace{1.5cm}
	\begin{tikzpicture}
	\node[hollow node] (0) at (0,0) {$X_1$};
	\node[hollow node] (1) at (1.65,0.7)  {$X_4$};
	\node[hollow node] (2) at (1.65,-0.7)  {$X_3$};
	\path (0) edge (1)  ;
	\path (0) edge (2);
	\path (0) -- (1) node [midway,auto=left] {\small $\beta_{14}$};
	\path (0) -- (2) node [midway,auto=right] {\small $\beta_{13}$};
	\end{tikzpicture} 
	\hspace{0.3cm}
	\begin{tikzpicture}
	\node[hollow node] (0) at (0,0) {$X_1$};
	\node[hollow node] (1) at (1.65,0.7)  {$X_4$};
	\node[hollow node] (2) at (1.65,-0.7)  {$X_3$};
	\path (0) edge (1)  ;
	\path (1) edge (2);
	\path (0) -- (1) node [midway,auto=left] {\small $\beta_{14}$};
	\path (1) -- (2) node [midway,auto=right] {\small $\beta_{34}$};
	\end{tikzpicture} 
	\hspace{0.3cm}
	\begin{tikzpicture}
	\node[hollow node] (0) at (0,0) {$X_1$};
	\node[hollow node] (1) at (1.65,0.7)  {$X_4$};
	\node[hollow node] (2) at (1.65,-0.7)  {$X_3$};
	\path (0) edge (2)  ;
	\path (1) edge (2);
	\path (0) -- (2) node [midway,auto=right] {\small $\beta_{13}$};
	\path (1) -- (2) node [midway,auto=left] {\small $\beta_{34}$};
	\end{tikzpicture} 
	\caption{The first tree from the left has four nodes where node~2 has a latent variable. If node~2 is suppressed, there are three possible ways to reconnect the three remaining nodes into a tree again. 
	}
	\label{fig:ignore_node}
\end{figure}

In this paper we do not modify the original tree but take the latent variables into account.
Let $\T = (V, E)$ be an undirected tree and consider the Hüsler--Reiss distribution~\eqref{eq:mvhr_short} with parameter matrix $\Lambda = \Lambda(\theta)$ in~\eqref{eqn:lambda}. When there are nodes with latent variables, the question is whether it is still possible to identify the $d-1$ free edge parameters $\theta_e$ from the distribution of the subvector of observable variables only. Let $U \subseteq V$ denote the set of indices of the observable variables. On the one hand, Eq.~\eqref{eq:XlimN} implies
\begin{equation} \label{eq:XUlimNU}
	(\ln X_v-\ln X_u)_{v\in U\setminus u} \mid X_u>t \dto
	\mathcal{N}_{|U\setminus u|}\bigl(
		\mu_{U,u}(\theta), 
		\Sigma_{U,u}(\theta)
	\bigr), 
	\qquad t \to \infty, 
\end{equation}
with $\mu_{U,u}(\theta)$ and $\Sigma_{U,u}(\theta)$ as in~\eqref{eqn:muY} and~\eqref{eqn:sigmaY} but with $V$ replaced by $U$. On the other hand, $\mu_{U,u}(\theta)$ and $\Sigma_{U,u}(\theta)$ together determine the stdf $l_U$ of the subvector $X_U$ in~\eqref{eq:stdf_hr} through the identities \eqref{eq:Sigma2Gamma} and~\eqref{eq:mu2lambda}. The question is thus whether the parameter vector $\theta$ is still identifiable from the $|U \setminus u|$-variate normal distributions on the right-hand side of~\eqref{eq:XUlimNU}, where $u$ ranges over $U$.

\paragraph{Example.} 
Let $X=(X_a,X_b,X_c)$ have unit-Pareto margins and suppose that $X \in D(H_{\Lambda(\theta)})$ where $\Lambda(\theta)$ is as in \eqref{eqn:lambda} with respect to the chain tree $\T$ with nodes $V = \{a, b, c\}$ and edges between $a$ and $b$ and between $b$ and $c$. Since a parameter is linked to each (undirected) edge of the graph, the parameter vector is $\theta=(\theta_{ab}, \theta_{bc})$. Suppose the variable $X_b$ is latent. By~\eqref{eq:XUlimNU} we have 
\[
    \ln X_c-\ln X_a \mid X_a>t
    \dto
    \mathcal{N}\bigl(-(\theta_{ab}^2 + \theta_{bc}^2)/2, (\theta_{ab}^2 + \theta_{bc}^2)\bigr),
    \qquad t \to \infty.
\]
It is clear that from the limiting normal distribution, we cannot identify $\theta_{ab}$ and $\theta_{bc}$.
\smallskip

In this section it is shown that as long as all nodes with missing variables have degree at least three, the parameters associated to the Hüsler--Reiss distribution of the full vector are still identifiable and hence there is no need to change the tree. To this end, note that by~\eqref{eqn:muY},~\eqref{eq:Sigma2Gamma} and~\eqref{eq:mu2lambda} with $V$ replaced by $U \subseteq V$ such that $u \in U$, the mean vectors $\mu_{U,u}(\theta)$ and covariance matrices $\Sigma_{U,u}(\theta)$ are determined completely by the path sums
\begin{equation} \label{eq:sumpath}
	p_{ab} = \sum_{e \in \path{a}{b}}\theta_e^2 = 4 \lambda_{ab}^2,
	\qquad a, b \in U,
\end{equation}
and that, vice versa, the values of these path sums are determined by the vectors $\mu_{U,u}(\theta)$ and the matrices $\Sigma_{U,u}(\theta)$. If we know the distribution of $X_U = (X_u, {u \in U})$, we can compute the values of these sums, and if we know these sums, we can compute the stdf $l_U$ of $X_U$. The question is thus whether or not the edge parameters $\theta_e$ are identifiable from the values of the path sums $p_{ab}$ for $a, b \in U$. According to the following proposition, there is a surprisingly simple criterion to decide whether this is the case or not. 

\begin{prop} \label{prop:identif}
	Let $\T = (V, E)$ be an undirected tree and let $X = (X_v, {v \in V})$ have unit Pareto margins and be in the max-domain of attraction of the structured Hüsler--Reiss distribution $H_\Lambda$ in~\eqref{eq:mvhr_short} with parameter matrix $\Lambda = \Lambda(\theta)$ in~\eqref{eqn:lambda}. Let $U \subseteq V$ be the set of nodes corresponding to the observable variables. The parameter vector $\theta$ is identifiable from $X_U = (X_u,{u \in U})$ if and only if every node $u \in V \setminus U$ has degree at least three.
\end{prop}

\begin{proof}
    \emph{Necessity.} Assume that the elements of the parameter $\theta\in (0,\infty)^{d-1}$ are uniquely identifiable. Let $\bar{U}=V\setminus U\neq \varnothing$ be the set of nodes with latent variables. We need to show that every $v \in \bar{U}$ has degree $d(v)$ at least $3$. We will do this by contraposition.
    As a tree is connected by definition, there cannot be a node of degree zero.
     
    First, assume there is $v\in \bar{U}$ such that $d(v)=1$. The node $v$ must be a leaf node, and in this case there is no path $\path{a}{b}$ with $a, b \in U$ that passes by $v$, and thus $\theta_{uv}^2$, with $u$ the unique neighbor of $v$, does never appear in the sum \eqref{eq:sumpath}. Hence $\theta_{uv}$ is not identifiable, which is a contradiction to the assumption.

    Second, assume there exists $v\in \bar{U}$ with $d(v)=2$. Then $v$ has exactly two neighbors, $i$ and $j$, say. Every path sum $p_{ab}$ for $a,b\in U$ will contain either the sum of the squared parameters, $\theta_{iv}^2 + \theta_{jv}^2$, or neither of these. Hence, the individual edge parameters $\theta_{iv}$ and $\theta_{jv}$ are not identifiable, yielding a contradiction. (This generalizes the example given before the statement of the proposition.)

    \emph{Sufficiency.} Assume that all nodes with latent variables are of degree three or more. Let $e = (u, v) \in E$. We will find a linear combination of the path sums \eqref{eq:sumpath} equal to $\theta_{uv}^2$. 
	
	If $u, v \in U$, then the one-edge path sum $p_{uv} = \theta_{uv}^2$ already meets the condition.
	
	Suppose that $u\in \bar{U}$. By assumption, $u$ has at least two other neighbors besides $v$, say $w$ and $x$. If $v\in U$, then put $\hat{v} = v$. Otherwise, start walking at $v$ away from $u$ until you encounter the first visible node, say $\hat{v} \in U$. There must always be such a node, since $V$ is finite and since all leaves are observable by assumption. Similarly, let $\hat{w} \in U$ and $\hat{x} \in U$ be the first visible nodes encountered when walking away from $u$ and starting in $w$ and $x$, respectively. Note that $\hat v$, $\hat w$, and $\hat x$ are all different since otherwise the graph would contain a non-trivial cycle, which is not possible in the case of a tree. We can thus observe the sums
	\begin{align*}
		p_{\hat{v}\hat{w}}
		&=
		p_{\hat{v}u} + p_{u\hat{w}}\, ,\\
		p_{\hat{v}\hat{x}}
		&=
		p_{\hat{v}u} + p_{u\hat{x}}\, ,\\
		p_{\hat{w}\hat{x}}
		&=
		p_{\hat{w}u} + p_{u\hat{x}}\, .
	\end{align*}
	Since $p_{yz} = p_{zy}$ for every $y,z\in V$, the previous identities constitute three linear equations in three unknowns that can be solved explicitly, producing the values of $p_{u\hat{v}},\, p_{u\hat{w}},\, p_{u\hat{x}}$. In particular, summing the first two equations, subtracting the third, and dividing by two, we find
	\[
		p_{u\hat{v}}
		=
		\tfrac{1}{2} p_{\hat{v}\hat{w}}
		+ \tfrac{1}{2} p_{\hat{v}\hat{x}}
		- \tfrac{1}{2} p_{\hat{w}\hat{x}}\, .
	\]
	
	If $v \in U$, then $v = \hat{v}$, and $\path{u}{v} = \{ e \}$, so that the above equation shows how to combine path sums in a linear way to extract $p_{uv} = \theta_{e}^2$.
		
	If $v \not\in U$, then we can repeat the same procedure with $u$ replaced by $v$. The result is a formula expressing $p_{v\hat{v}}$ as a linear combination of three visible path sums. Now since
	\[
		\theta_{e}^2 
		= p_{u\hat{v}} - p_{v\hat{v}}\, ,
	\]
	we have found a way to extract $\theta_{e}^2$ by a linear combination of at most six visible path sums.
\end{proof}

The proof of Proposition~\ref{prop:identif} consists in solving the equations \eqref{eq:sumpath} with $p_{ab}$ as known and $\theta_e^2$ as unknown. Clearly, this is a linear system of equations and the question is thus whether the coefficient matrix defining the system has full column rank. It is an open question how to write down this matrix, which contains only zeroes and ones, in terms of the tree's adjacency matrix in such a way that an algebraic criterion on the latter matrix can be formulated.

The identifiability criterion in Proposition~\ref{prop:identif} allows nodes with latent variables to be adjacent and still counting in the computation of each other's degree. Consider for instance the following tree:
\begin{center}
\begin{tikzpicture}[scale=0.5]
	\node[hollow node] (1) at (4, 1.73) {$X_1$};
	\node[hollow node] (2) at (4, -1.73) {$X_2$};
	\node[hollow node] (3) at (-4, -1.73) {$X_3$};
	\node[hollow node] (4) at (-4, 1.73) {$X_4$};
	\node[hollow node, fill=gray] (5) at (1.5, 0) {\textcolor{white}{5}};
	\node[hollow node, fill=gray] (6) at (-1.5, 0) {\textcolor{white}{6}};
	\path (1) edge (5);
	\path (1) -- (5) node [midway,auto=right] {\small $\theta_{15}$};
	\path (2) edge (5);
	\path (2) -- (5) node [midway,auto=left] {\small $\theta_{25}$};
	\path (3) edge (6);
	\path (3) -- (6) node [midway,auto=right] {\small $\theta_{36}$};
	\path (4) edge (6);
	\path (4) -- (6) node [midway,auto=left] {\small $\theta_{46}$};
	\path (5) edge (6);
	\path (5) -- (6) node [midway,auto=right] {\small $\theta_{56}$};
\end{tikzpicture}
\end{center}
The variables at the adjacent nodes $5$ and $6$ are latent. Both nodes have degree three and each of the five edge parameters $\theta_e$ can be solved from the path sums $p_{ab}$ between nodes $a, b \in \{1,\ldots,4\}$.

The previous example may give the impression that for the identifiability criterion to hold it is actually enough that all variables on leaf nodes are observable. Although the latter property is indeed necessary, it is not sufficient, as illustrated by the example before Proposition~\ref{prop:identif}.

\section{Estimation} \label{section:estim}

Let $\T = (V, E)$ be an undirected tree with nodes $V = \{1,\ldots,d\}$ and let $(\xi_{v,i},{v \in V}, i=1,\ldots, n)$ be an independent random sample from the distribution of $\xi$ satisfying the assumptions in Section~\ref{ssec:introX}.
Further, let $U \subseteq V$ be the set of indices of observable variables and assume that every $u \in V \setminus U$ has degree at least three, so that, by Proposition~\ref{prop:identif}, the Hüsler--Reiss edge parameters $\theta = (\theta_e, e \in E)$ in the definition of $\Lambda(\theta)$ in \eqref{eqn:lambda} are identifiable from the distribution of the subvector $\xi_U = (\xi_v, v \in U)$.

We propose three methods for estimating the parameter vector $\theta$. The first one, called moment estimator (Section~\ref{ssec:MME}), builds upon the one introduced in \citet{engelke}. The second estimator comes from the optimization of a composite likelihood function (Section~\ref{ssec:cle}). The third estimator, finally, is based on bivariate extremal coefficients (Section~\ref{ssec:ece}) and on the method in \citet{eks16}. 
All estimators are functions of the subvectors $(\xi_{U,i}) = (\xi_{v,i},{v \in U})$ for $i=1,\ldots, n$ only.

An important remark for this whole section is related to the fact that $X$ as introduced in Section~\ref{ssec:introX} should have unit Pareto margins, obtained after the transformation $X_v=1/(1-F_v(\xi_v))$ where $F_v$ is the marginal distribution function of $\xi_v$ for $v \in V$. It is unrealistic to assume that the functions $F_v$ are known, so in practice we use their empirical versions, $\hat{F}_{v,n}(x)=\big[\sum_{i=1}^n\mathbbm{1}(\xi_{v,i}\leq x)\big]/(n+1)$. The estimates of the edge parameters will then be based upon the sample $\hat{X}_1, \ldots, \hat{X}_n$ with coordinates
\begin{equation*} 
	\hat{X}_{v,i} = \frac{1}{1-\hat{F}_{v,n}(\xi_{v,i})}\,,\qquad 
	v \in U, \quad i = 1, \ldots, n,
\end{equation*}
considered as a random sample from the distribution of $X_U = (X_u,{u \in U})$.

A variable indexed by the double subscript $W,i$ will denote the $i$-th observation of variables on nodes belonging to the set $W \subseteq U$: for instance $\hat{X}_{W,i}=(\hat{X}_{v,i}\,, v\in W)$. Such vectors are taken to be column vectors of length $|W|$. When $W = U$ we just write $\hat{X}_i$.

\subsection{Method of moments estimator}
\label{ssec:MME}

\citet{engelke} introduce an estimator of the matrix $\Lambda$ of the Hüsler--Reiss distribution, based on sample counterparts of the matrices $\Gamma_{W,u}(\Lambda)$ in~\eqref{eq:hrdist}. Relying on \eqref{eq:Sigma2Gamma} with $V$ replaced by $W \subseteq U$, we will apply their method to the vector of observable variables and then add a least-squares step to extract the edge parameters $\theta_e$.

As a starting point we take the result in~\eqref{eq:XUlimNU} and as suggested by \citet{engelke} for given $k\in \{1,\ldots n\}$ we obtain the log-differences
\begin{equation}
\label{eq:Delta}
	\Delta_{uv,i} =
	\ln\hat{X}_{v,i}-\ln\hat{X}_{u,i}\, ,
\end{equation}
for $u, v \in U$ and for $i \in I_{u} = \{i = 1,\ldots,n: \hat{X}_{u,i} > n/k\}$.
The proposed estimators of $\mu_{U,u}$ and $\Sigma_{U,u}$ are respectively the sample mean vector
\begin{equation*} 
    \hat{\mu}_{U,u}
    =
    \frac{1}{|I_u|}\sum_{i\in I_u}(\Delta_{uv,i}, v\in U \setminus u)
\end{equation*}
and the sample covariance matrix
\begin{equation*}
    \hat{\Sigma}_{U,u}
    =
    \frac{1}{|I_u|}\sum_{i\in I_u}(\Delta_{uv,i}-\hat{\mu}_{U,u}, v\in U\setminus u)
    (\Delta_{uv,i}-\hat{\mu}_{U,u}, v\in U\setminus u)^\top\, .
\end{equation*}
To estimate the vector of edge parameters $\theta=(\theta_e, e \in E)$, we propose the least squares estimator
\begin{equation} \label{eqn:510} 
	\hat{\theta}^{\mathrm{MM}}_{n,k}
	=
	\argmin_{\theta\in (0,\infty)^{E}}
	\sum_{u \in U}
	\|\hat{\Sigma}_{U,u}-\Sigma_{U,u}(\theta)\|_F^2\, .
\end{equation}
where $\|\,\cdot\,\|_F$ is the Frobenius norm. In this way, we take advantage of the empirical covariance matrices $\hat{\Sigma}_{U,u}$ for each $u \in U$ and thus of each exceedance set $I_u$.

In~\eqref{eqn:510}, for each $u \in U$, we consider the covariance matrix of the log-differences $\Delta_{uv,i}$ for all $v \in U \setminus u$. However, if $v$ is far away from $u$ in the tree, then the extremal dependence between $\xi_u$ and $\xi_v$ may be weak and the difference $\Delta_{uv,i}$ may carry little information. Therefore, we propose a modified estimator where, for each $u \in U$, we limit the scope to a subset $W_u \subseteq U$ of observable variables indexed by nodes near $u$, producing the estimator 
\begin{equation} \label{finalmm}
	\hat{\theta}^{\mathrm{MM}}_{n,k}
	=
	\argmin_{\theta\in(0,\infty)^{E}}
	\sum_{u\in U} \| \hat{\Sigma}_{W_u, u}-\Sigma_{W_u,u}(\theta) \|_F^2\, .
\end{equation}
Besides being simpler to compute, the modified estimator~\eqref{finalmm} performed better than the one in~\eqref{eqn:510} in Monte Carlo experiments. One possible explanation is that by excluding pairs with weak extremal dependence, the bias of the estimator diminishes.

When choosing the sets $W_u$, care needs to be taken that the parameter vector $\theta$ is still identifiable from the collection of covariance matrices $\Sigma_{W_u,u}(\theta)$ for $u \in U$. 
The set of path sums $p_{ab}$ for $a, b \in U$ in Proposition~\ref{prop:identif} is now reduced to the set of the path sums $p_{ab}$ for $a, b \in W_u$ and $u \in U$. Whether or not these are still sufficient to identify $\theta$ needs to be checked on a case-by-case basis. This issue is illustrated in Appendix~\ref{app:Wu}.

\subsection{Composite likelihood estimator} \label{ssec:cle}

The composite likelihood estimator (CLE) is again based on the result in \eqref{eq:XUlimNU}. This time however we maximize a composite likelihood function with respect to the parameter $\theta$ directly. The composite likelihood function consists of multiplication of likelihoods which are defined on subtrees. 

As for the method of moments estimator in Section~\ref{ssec:MME}, we consider for each $u \in U$ a set $W_u \subseteq U$ of nodes that are close to $u$ in the tree, taking care to include sufficiently many variables so that the edge parameters are still identifiable (Appendix~\ref{app:Wu}). Recall the log-differences $\Delta_{uv,i}$ in \eqref{eq:Delta} and the exceedance set $I_u$ right below \eqref{eq:Delta}. Let $\phi_p(\,\cdot\,;\Sigma)$ be the density function of the centered $p$-variate normal distribution with covariance matrix $\Sigma$. The composite likelihood estimator $\hat{\theta}_{n,k}^{\mathrm{CLE}}$ is the maximizer of the composite likelihood
\[
    L\big(\theta; \, \{\Delta_{uv,i}: v\in W_u\setminus u,\, i\in I_u, u\in U\}\big)
    \\=
    \prod_{u\in U}\prod_{i\in I_u}
    \phi_{|W_u \setminus u|}\bigl(
    	(\Delta_{uv,i})_{v\in W_u} - \mu_{W_u, u}(\theta);
    	\Sigma_{W_u, u}(\theta)
   \bigr).
\]
We aggregate the likelihoods of the different normal distributions for all $u \in U$ treating the samples of log-differences as independent, although they are not. Results from Monte Carlo simulation experiments (Appendix~\ref{app:simu}) show that the performance of the CLE is comparable to the one of the moment estimator and the extremal coefficient estimator. 

Other estimation methods based on locally defined likelihoods are used by \citet{engelke+h:2020} and \citet{joe}. The method of \citet{engelke+h:2020} estimates the parameters associated to each clique separately. For trees this means that there are $d-1$ one-variate likelihood functions to optimize, a problem which is doable even in trees with many nodes. A problem with this estimator is that it is inapplicable if there are latent variables because there will always be an adjacent pair of variables with one of them being an unobservable, and making it impossible to estimate the corresponding edge parameter. 
The estimator of \citet{joe} is based on pairwise likelihoods, which can be any pairs, not only adjacent pairs as in the estimator of \citet{engelke+h:2020}. It is obtained by optimizing the composite likelihood which consists of multiplying the pairwise likelihoods. This estimator is applicable when there are latent variables as long as all possible pairs between the observed variables are included in the composite likelihood function. It is close in spirit to the pairwise extremal coefficients estimator considered next.

\subsection{Pairwise extremal coefficients estimator} \label{ssec:ece}

The pairwise extremal coefficients estimator (ECE), defined for general tail dependence models in \citet{eks16}, is based on the bivariate stable tail dependence function (stdf) in \eqref{eq:stdf_hr}. It minimizes the weighted distance between a non-parametric estimate and the fitted parametric stdf.

Let $l$ be the stdf in \eqref{eq:stdf} and recall that the extremal coefficient associated to a node set $J \subseteq V$ is defined as 
\begin{equation} 
\label{eqn:lJ}
	l(x_J) = l_J(1, \ldots, 1) = \lim_{t \to \infty} t \, \P\left( \max_{j \in J} X_j > t \right),
\end{equation}
where $x_J = (\mathbbm{1}_{\{j \in J\}}, j \in V)$ and where $l_J$ is the stdf of the subvector $X_J$. For the Hüsler--Reiss distribution with parameter matrix $\Lambda$ and for a pair of nodes $J = \{u, v\}$, the bivariate extremal coefficient is just $l_J(1, 1) = 2 \Phi(\lambda_{uv})$, with $\Phi$ the standard normal cdf. In case $\Lambda = \Lambda(\theta)$ in \eqref{eqn:lambda}, the pairwise extremal coefficient depends on the path sum $p_{uv} = \sum_{e \in \path{u}{v}} \theta_e^2$ via
\begin{equation}
\label{eq:EC2}
	l_J(1, 1; \theta) = 2 \Phi(\sqrt{p_{uv}}/2), \qquad J = \{u, v\}.
\end{equation}

The non-parametric estimator of the stdf dates back to \citet{drees1998} and yields the following estimator for the extremal coefficient $l_J(1, \ldots, 1)$ for $J \subseteq V$:
\begin{equation} 
\label{eqn:lJkn}
  \hat{l}_{J;n,k}(1,\ldots,1)
  =
  \frac{1}{k}\sum_{i=1}^n
  \mathbbm{1}\left( \max_{j \in J} n\hat{F}_{j,n}(\xi_{j,i}) >n+1/2-k \right).
\end{equation}

Let $\mathcal{Q} \subseteq \{ J \subseteq U : |J| = 2 \}$ be a collection of pairs of nodes associated to observable variables and put $q = |\mathcal{Q}|$, ensuring that $q \ge |E| = d-1$, the number of free edge parameters. The pairwise extremal coefficients estimator (ECE) of $\theta$ is
\begin{equation} 
\label{eqn:ECE}
	\hat{\theta}^{\mathrm{ECE}}_{n,k}
	=
	\argmin_{\theta \in (0,\infty)^{E}}
	\sum_{J \in \mathcal{Q}} \left( \hat{l}_{J;n,k}(1,1) - l_J(1, 1;\theta) \right)^2.
\end{equation}    

If $\mathcal{Q}$ is the collection of all possible pairs of nodes in $U$, then the pairwise extremal coefficients \eqref{eq:EC2} give us access to all path sums $p_{ab}$ for $a, b \in U$, and Proposition~\ref{prop:identif} guarantees we can identify $\theta$. If, however, $\mathcal{Q}$ is a smaller set of pairs, then the identifiability of $\theta$ from the resulting path sums needs to be checked on the case at hand.

\section{High water levels on the Seine network}
\label{sec:Seine}

We have chosen to present an application that allows us to demonstrate the identifiability criterion outlined in Section~\ref{sec:latent}. Data were collected from http://www.hydro.eaufrance.fr, a web-site of the french Ministry of Ecology,
Energy and Sustainable Development, and span the period from January 1987 to April 2019 with gaps for some of the measurement stations. The data represent water levels, in cm, at five locations on the Seine river: Paris, Meaux, Melun, Nemours and Sens. The map on Figure~\ref{fig:seine} shows part of the actual Seine network. The schematic representation of the graphical model used in the estimation is shown in Figure~\ref{fig:seine_scheme}. The tree has $d = 7$ nodes, two of which are associated to latent variables. Since both these nodes have degree equal to three, Proposition~\ref{prop:identif} guarantees we can still identify all six edge parameters $\theta_1, \ldots, \theta_6$. For more information on the data set, some summary statistics and details on data preprocessing, we refer to Appendix~\ref{app:Seine:preproc}.

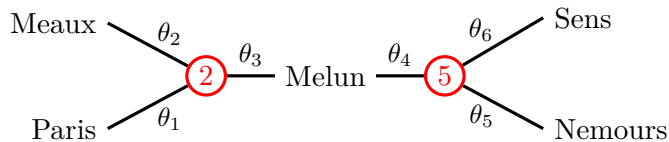
\begin{figure}[h]
\centering
\begin{tikzpicture}[scale=0.7]
\node at (0,0) (mel) {Melun};
\draw [very thick, -] (mel) -- ++(1.85cm,0)
       node[right, hollow node, red] (5) {$5$};
\draw [very thick, -] (mel) -- ++(-1.85cm,0)
       node[left, hollow node, red] (2) {$2$};
\draw [very thick, -] (5) -- ++(1.85cm,-1)
    node[right]  (nem) {Nemours};
\draw [very thick, -] (5) -- ++(1.85cm,+1.1)
        node[right]  (sens) {Sens};
\draw [very thick, -] (2) -- ++(-1.85cm,-1)
        node[left]  (paris) {Paris};
\draw [very thick, -] (2) -- ++(-1.85cm,1)
        node[left]  (me) {Meaux};
\path (paris) -- (2) node [midway,auto=right] {\small $\theta_1$};
\path (2) -- (me) node [midway,auto=right] {\small $\theta_{2}$};
\path (2) -- (mel) node [midway,auto=left] {\small $\theta_{3}$};
\path (mel) -- (5) node [midway,auto=left] {\small $\theta_{4}$};
\path (5) -- (nem) node [midway,auto=right] {\small $\theta_{5}$};
\path (5) -- (sens) node [midway,auto=left] {\small $\theta_{6}$};
\end{tikzpicture}
\caption{The Seine network with the tail dependence parameters associated to each edge of the tree.}
\label{fig:seine_scheme}
\end{figure}

\subsection{Estimates and confidence intervals}
   
We used all three estimators in Section~\ref{section:estim} to obtain estimates of the six parameters of extremal dependence. For the pairwise extremal coefficient estimator (ECE) it is possible to calculate standard errors thanks to the asymptotic distribution derived in \citet[Theorem~2.2]{eks16}. Computational details for the standard errors follow in Appendix~\ref{app:ECE}. The distributions of the MME and CLE are not known so we computed bootstrapped confidence intervals, known as basic bootstrap confidence intervals \citep[Chapter~2]{davison_hinkley_1997}, by resampling from the data.

The EC estimates and their 95\% confidence intervals are displayed in Figure~\ref{fig:eks_ci} for two of the parameters, namely $\theta_1$ and $\theta_4$. The confidence intervals using the MME and CLE are narrower as can be seen from Figure~\ref{fig:all_theta_ci}. The plots for $\theta_2, \theta_5, \theta_6$ are similar to the one for $\theta_1$: the 95\% confidence intervals never include zero,  suggesting that the extremal dependence between the corresponding variables is not perfect and hence that the edges cannot be collapsed. In Section~\ref{sec:latent}, we alluded to the possibility of circumventing the issue of latent variables by suppressing nodes and redrawing edges. The fact that the confidence intervals do not include zero indicate that doing so would have produced a misleading picture of extremal dependence. 

The plot of $\theta_3$, similarly to the plot of $\theta_4$, does contain a segment over $k$ where the lower confidence bound reaches zero: for $\theta_4$ this is approximately $k\in[260,360]$, while for $\theta_3$ it is $k\in[90,180]$. Although the confidence intervals for $\theta_3$ and $\theta_4$ indicate some instability of the estimated parameters, we believe that collapsing the edges is not advisable, especially in networks with many more unobservable variables. Moreover, the river distance, which is one of the important factors in tail dependence \citep{asadi}, is rather long between node~2 and Melun and between Melun and node~5, so that there is no physical motivation for collapsing the corresponding edges.

\begin{figure}
    \centering
    \includegraphics[scale=0.7]{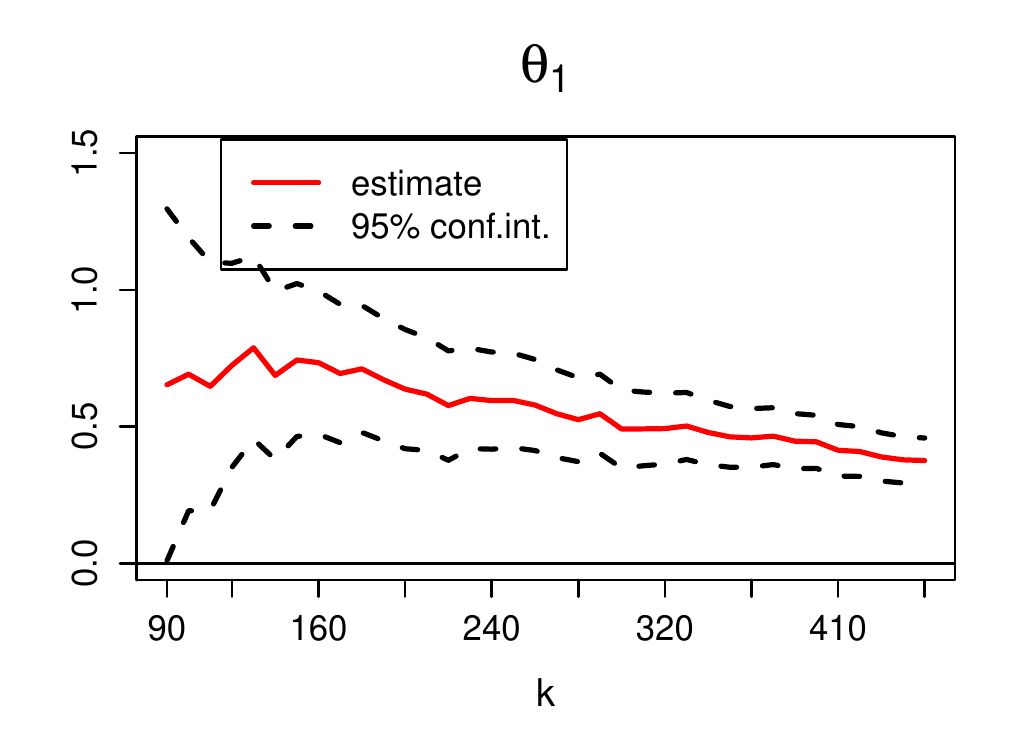}
    \includegraphics[scale=0.7]{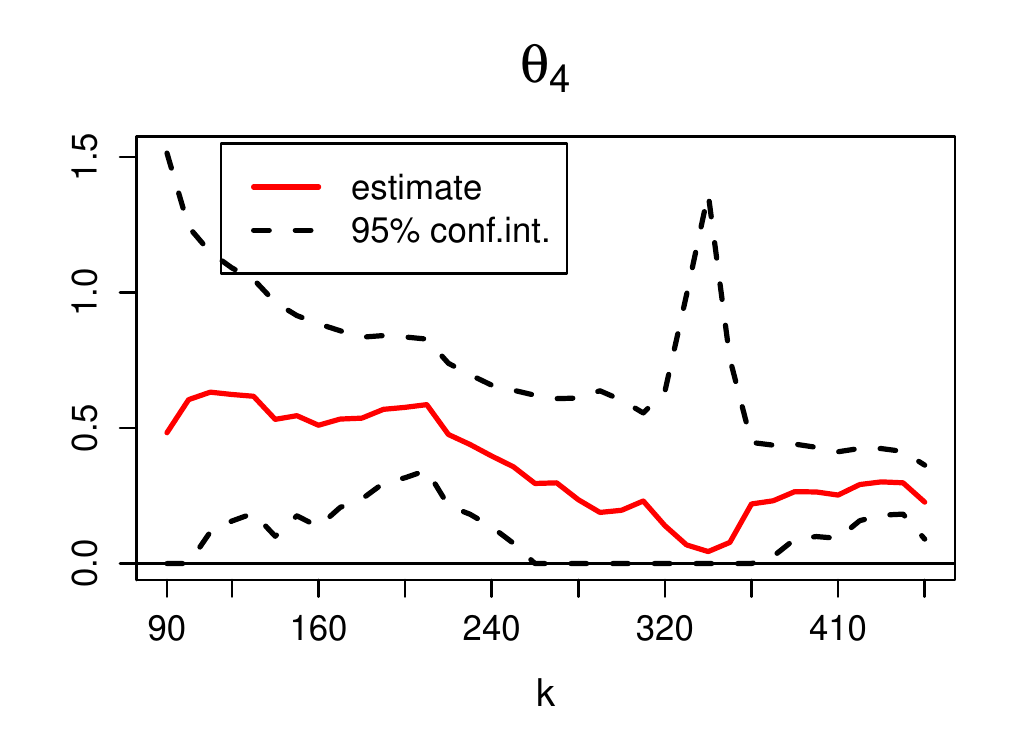}
    \caption{Point estimates and confidence intervals for the pairwise ECE.}
    \label{fig:eks_ci}
\end{figure}

For a point estimate per parameter we need to average out over a range of $k$. The chosen range per estimator and per parameter need not be the same. As a rule we select a range around the beginning where the estimates start stabilizing around a certain level, omitting the most volatile part for relatively small $k$. Most of the time we thus consider $k\in [100,220]$. In this way we end up with the point estimates displayed for comparison in Figure~\ref{fig:all_theta_ci}.

\begin{figure}[ht]
	\centering
	\includegraphics[scale=0.9]{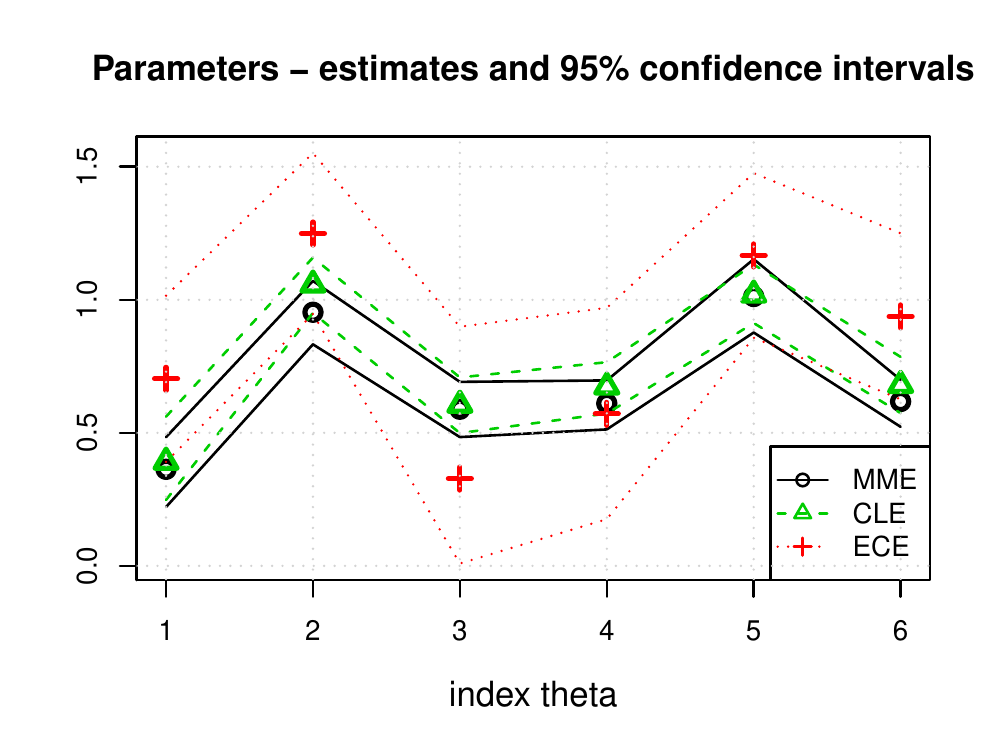}
	\caption{Parameters -- estimates and confidence intervals. The confidence intervals of the moment and composite likelihood estimators are bootstrapped, namely $\theta\in [2\hat{\theta}-q^*_{0.975}, 2\hat{\theta}-q^*_{0.025}]$, where $q^*_{\alpha}$ is the $\alpha$-quantile of the bootstrapped distribution of $\hat{\theta}$.}
	\label{fig:all_theta_ci}
\end{figure}


Given the similarities between the MME and CLE, the estimates are pooled in an average of the two for each parameter. 

\subsection{Considerations on the goodness-of-fit of the model}
\label{ssec:gof}

The Hüsler--Reiss family would not be an appropriate extremal dependence model if some of the variables would exhibit asymptotic independence. As kindly suggested by a Reviewer, we compared non-parametric estimates of multivariate tail dependence coefficients $\P(\min_{v \in W} X_v > t \mid X_u > t) = t \, \P(\min_{v \in W \cup u} X_{v} > t)$ for $W \subsetneq U$ and $u \in U \setminus W$ at finite thresholds $t = n/k$ with their postulated limits as $t \to \infty$ based on the fitted Hüsler--Reiss stdf. The results (not shown) supported the hypothesis of asymptotic dependence for nearly all subvectors $X_{W \cup u}$ of variables.

To assess how well the model from Section~\ref{ssec:introX} fits the data, we compare non-parametric and model-based estimates of quantities describing extremal dependence, such as pairwise and triple-wise extremal coefficients and the Pickands dependence function. For $J \subseteq U$, recall the extremal coefficient $l_J(1, \ldots, 1)$ in \eqref{eqn:lJ} and its non-parametric estimate $\hat{l}_{J;n,k}(1,\ldots,1)$ in \eqref{eqn:lJkn}, also called empirical extremal coefficient. The extremal coefficient $l_J(1, \ldots, 1)$ is always between $1$ and $|J|$, corresponding to perfect extremal dependence and to extremal independence, respectively.  

Figure~\ref{fig:extrcoef} compares the model-based extremal coefficients obtained from \eqref{eq:stdf_hr} by plugging in parameter estimates with the empirical counterparts for pairs and triples $J \subseteq U$. At least visually the fit is quite good for both estimators considered, which are the average of the CLE and MME on the one hand and the ECE on the other hand. Note that the ECE in \eqref{eqn:ECE} is constructed explicitly to ensure that the model-based pairwise extremal coefficients fit the empirical ones as closely as possible. It is therefore only natural that the extremal coefficients based on the ECE fit the empirical ones best. A more comprehensive comparison of the finite-sample performance of MME, CLE and ECE is reported in a simulation study in Appendix~\ref{app:simu}.

It should be noted that it is impossible to compute empirical extremal coefficients involving latent variables. Model-based estimates of such extremal coefficients can still be computed however, thanks to the identifiability of the parameter vector $\theta$.

\begin{figure}
    \centering
    \includegraphics[scale=0.65]{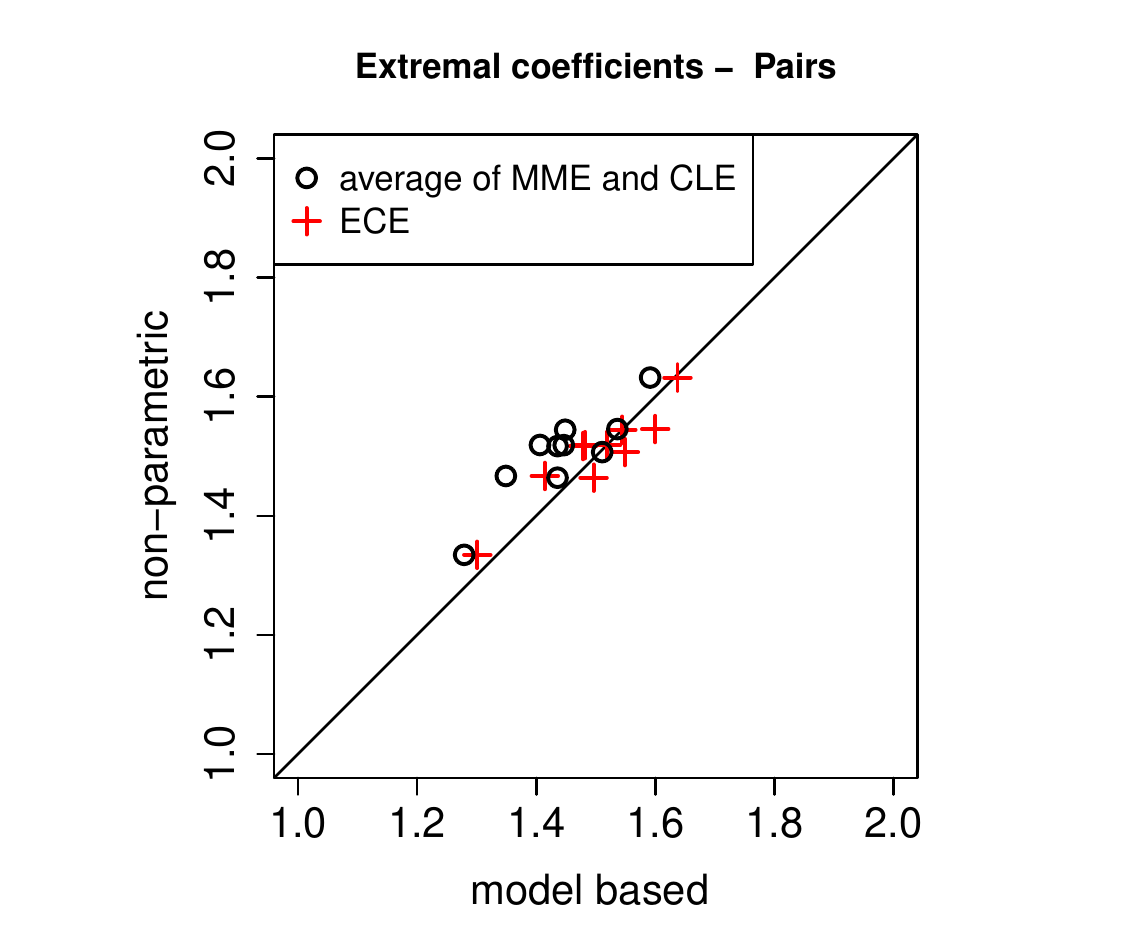}
    \includegraphics[scale=0.65]{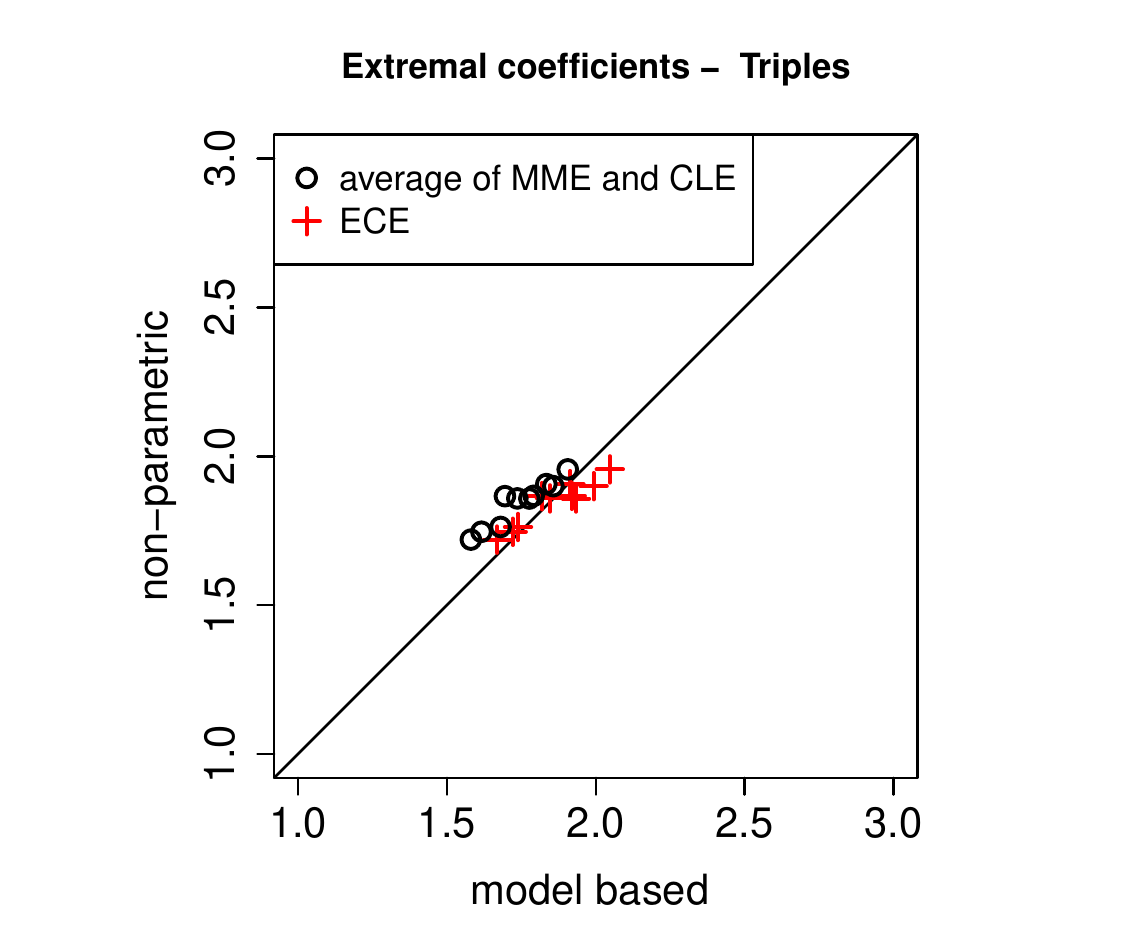}
    \caption{Non-parametric vs model-based extremal coefficients for pairs (left) and triples (right). }
    \label{fig:extrcoef}
\end{figure}

As another visual check of the goodness-of-fit of the assumed model we consider the bivariate Pickands dependence function, usually denoted by $A(w)$ for $w\in [0,1]$. For the Hüsler--Reiss extreme-value distribution at the pair $J = \{u, v\}$, it is equal to 
\begin{align*}
    A_{u,v}(w;\theta)&=l_J(1-w, w;\theta)
    \nonumber\\&=
    (1-w) \, \Phi\left(
    	\frac{\ln(\frac{1-w}{w}) + \frac{1}{2} p_{uv}}{\sqrt{p_{uv}}}
    \right)
    +
    w \, \Phi\left(
    	\frac{\ln(\frac{w}{1-w}) + \frac{1}{2} p_{uv}}{\sqrt{p_{uv}}}
    \right)\, ,
\end{align*}
with $p_{uv}$ as in \eqref{eq:sumpath}. Hence the model-based estimator of $A_{u,v}(w;\theta)$ is $A_{u,v}(w;\hat{\theta}_{n,k})$ where $\hat{\theta}_{n,k}$ can be the average MME/CLE or the ECE.

The non-parametric counterpart of the Pickands dependence function is
\begin{equation*}
    \hat{A}_{u,v}(w)
    =
    \frac{1}{k}\sum_{i=1}^n
    \mathbbm{1} \{
    	n\hat{F}_{u,n}(\xi_{u,i})>n-k(1-w)+1/2
    	\text{ or }
    	n\hat{F}_{v,n}(\xi_{v,i})>n-kw+1/2
    \}.
\end{equation*}
The model-based Pickands dependence function is compared to the empirical counterpart in Figure~\ref{fig:pdf}. The plot is complemented with non-parametric 95\% confidence intervals for $A(w)$ computed by the bootstrap method introduced in \citet[Section~5]{kiril_seg_taf}. The general idea of the method is to approximate the distribution of $\sqrt{k}(\hat{l}_{n,k}-l)$ by the distribution of $\sqrt{k}(\hat{l}_{n,k}^\ast-\hat{l}_{n,k}^{\beta})$ where $\hat{l}_{n,k}^\ast$ is the empirical stable tail dependence function based on the ranks of a sample of size $n$ from the empirical beta copula and $\hat{l}_{n,k}^{\beta}$ is the stdf based on the empirical beta copula using the ranks of the original sample $(\xi_{v,i}, v \in U)$ for $i=1,\ldots, n$. A detailed description of the derived bootstrap confidence intervals is provided in Appendix~\ref{app:stdf:CI}.  

\begin{figure}
	\begin{center}
	\begin{tabular}{@{}c@{}c}
    \includegraphics[width=0.49\textwidth]{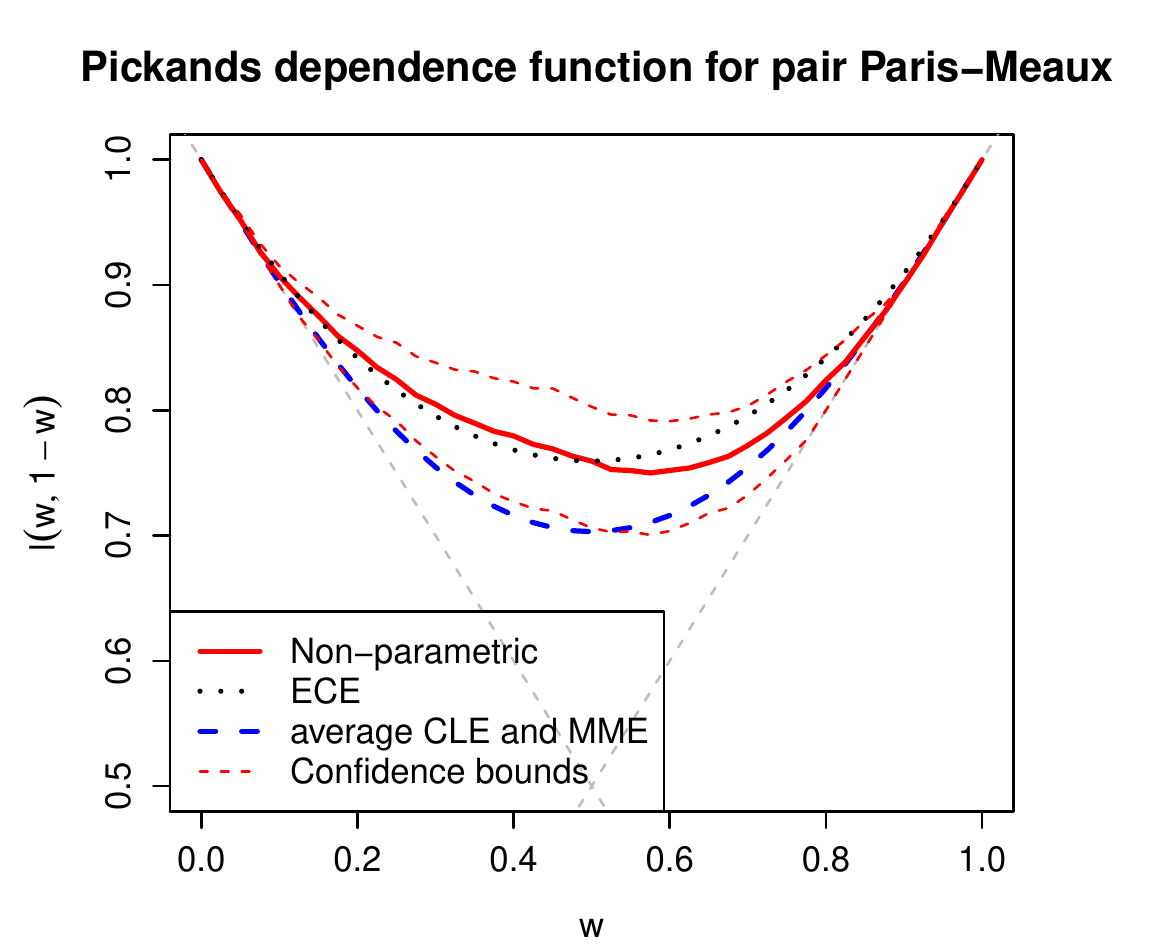}
    \includegraphics[width=0.49\textwidth]{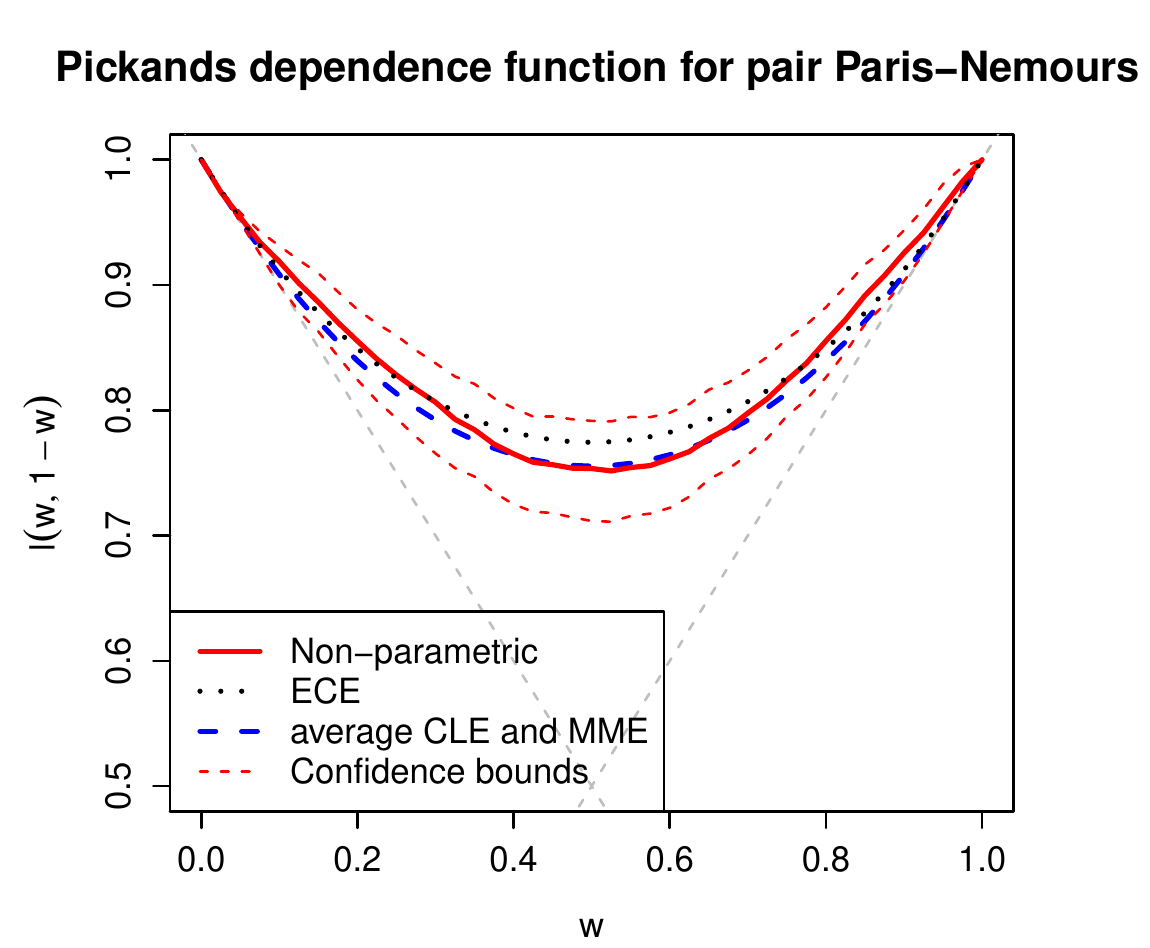}
    \end{tabular}
    \end{center}
    \caption{The empirical and model-based Pickands dependence function computed using the pooled CL and MM estimates and the EC estimates. The dashed gray lines show the lower limit $\max(w, 1-w)$ of any Pickands dependence function $A(w)$.}
    \label{fig:pdf}
\end{figure}

\subsection{Flow-connectedness and tail dependence}
\label{ssec:flow}

In a study by \citet{asadi} of data from the Danube, it was found that a key factor for extremal dependence between two locations is whether or not they are flow connected. Two locations are flow connected if one of them is downstream of the other one. Flow connectedness often dominates river distance or Euclidean distance in importance: variables on distant nodes that are flow connected might have stronger tail dependence than variables on nodes that are nearby but not flow connected. 

This effect is confirmed in our data too and is illustrated in Figure~\ref{fig:pair_plots}. The cities of Sens and Nemours are not flow connected but the Euclidean and river distance between them is smaller than the one between the flow connected cities of Sens and Paris. Still, the tail dependence seems to be stronger for the flow connected pair of locations.

\begin{figure}
    \centering
    \includegraphics[scale=0.85]{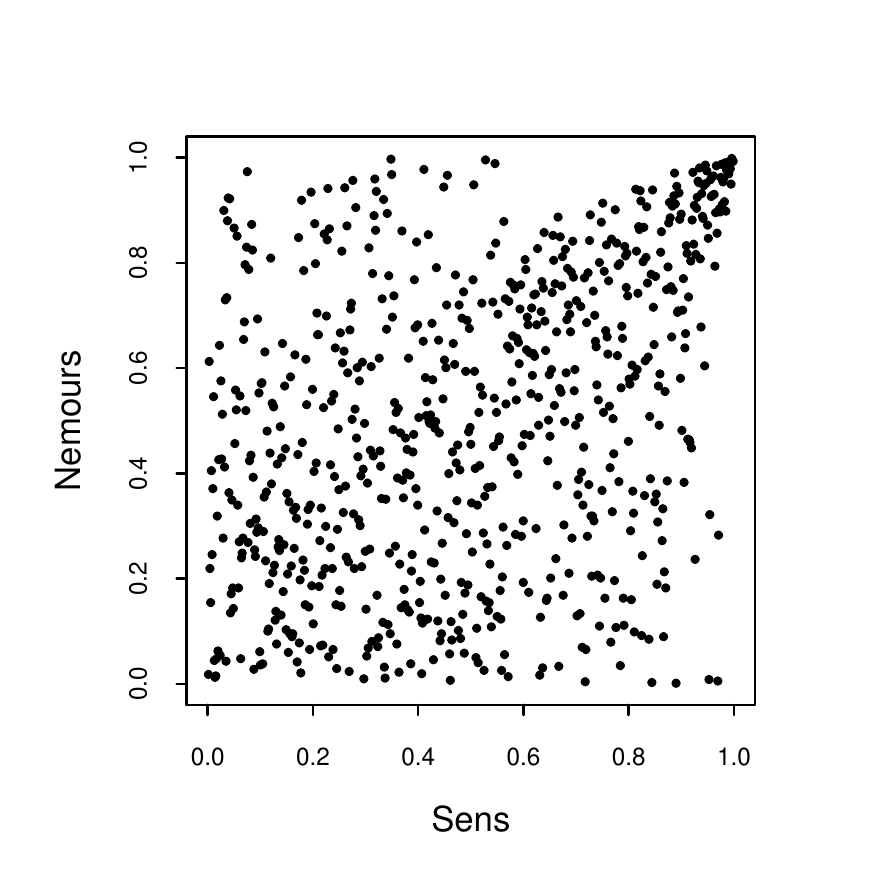}
    \includegraphics[scale=0.85]{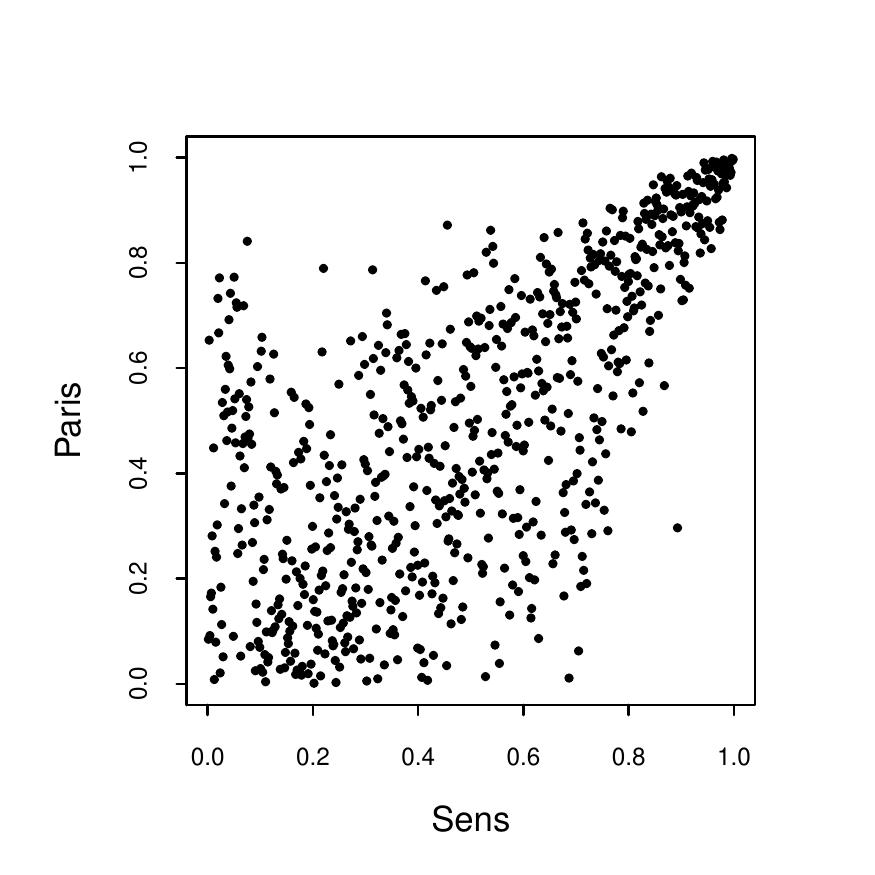}
    \caption{Scatterplots of uniform transformed data, $\hat{F}_{v,n}(\xi_{v,i}), v\in U, i=1, \ldots, n$, for two pairs of locations. Sens and Nemours (left) are not flow connected while Sens and Paris (right) are flow connected. It can be seen from the Seine map in Figure~\ref{fig:seine} that the river and Euclidean distance from Sens to Nemours is much smaller than the one from Sens to Paris. However the tail dependence seems to be stronger for the second pair of locations.}
    \label{fig:pair_plots}
\end{figure}

Figure~\ref{fig:heat_map} illustrates the tail dependence in the Seine network through a heat map of the pairwise extremal coefficients. Pairs which are flow connected are indeed the ones with stronger tail dependence (smaller extremal coefficient). According to both estimators the strongest tail dependence is to be found between Paris and the locations at node~2, node~5 and Melun.

\begin{figure}
	\centering
	\includegraphics[scale=0.75]{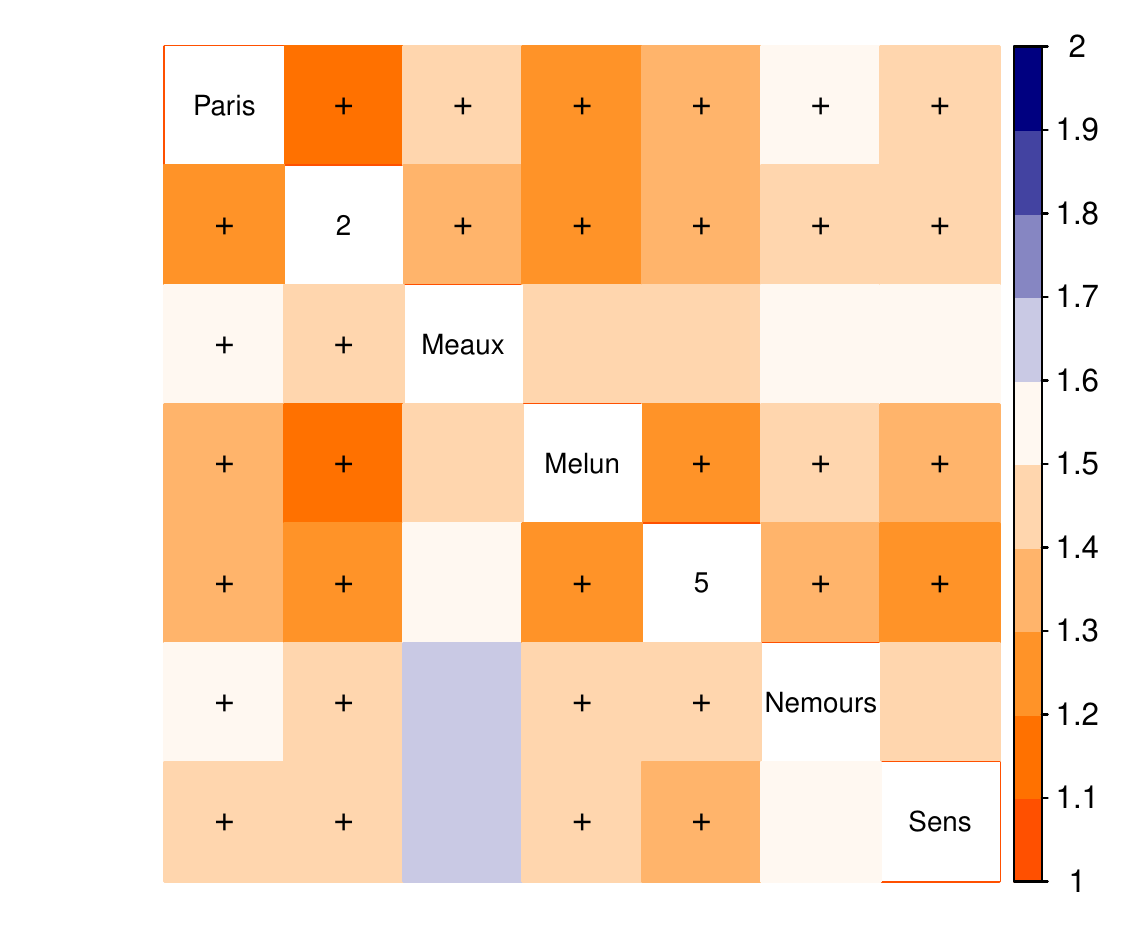}
	
	\caption{Heat map of the extremal coefficients. The upper diagonal is computed using the pooled MM and CL estimates and the lower diagonal uses the EC estimates. The crosses denote flow connected nodes.}
	\label{fig:heat_map}
\end{figure}

\subsection{Suppressing latent variables}
\label{ssec:suppress}

In Section~\ref{sec:latent} we alluded to the possibility of suppressing nodes with latent variables. Here we illustrate that method and compare the results with those presented so far. After removing nodes $2$ and $5$ from the Seine graph in Figure~\ref{fig:seine_scheme}, there is no unique way of reconnecting the remaining five nodes into a tree. Two possible structures for the reduced Seine graph are presented in Figure~\ref{fig:seine_scheme_shrink}. We opt for the right-hand graph and refer to the model associated to that tree as model~B. Model~A will refer to the one associated to the original graph in Figure~\ref{fig:seine_scheme}.

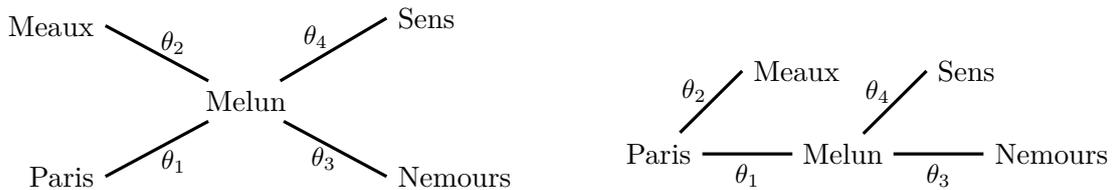
\begin{figure}[h]
	\centering
	\begin{tikzpicture} 
	\node at (0,0) (mel) {Melun};
	\draw [very thick, -] (mel) -- ++(1.85cm,-1)
	node[right]  (nem) {Nemours};
	\draw [very thick, -] (mel) -- ++(1.85cm,+1.1)
	node[right]  (sens) {Sens};
	\draw [very thick, -] (mel) -- ++(-1.85cm,-1)
	node[left]  (paris) {Paris};
	\draw [very thick, -] (mel) -- ++(-1.85cm,1)
	node[left]  (me) {Meaux};
	\path (paris) -- (mel) node [midway,auto=right] {\small $\theta_1$};
	\path (me) -- (mel) node [midway,auto=left] {\small $\theta_{2}$};
	\path (nem) -- (mel) node [midway,auto=left] {\small $\theta_{3}$};
	\path (sens) -- (mel) node [midway,auto=right] {\small $\theta_{4}$};
	\end{tikzpicture}
	\hspace{1cm}
	\begin{tikzpicture} 
	\node at (0,0) (mel) {Melun};
	\draw [very thick, -] (mel) -- ++(1.85cm,0)
	node[right]  (nem) {Nemours};
	\draw [very thick, -] (mel) -- ++(1.1cm,+1.1)
	node[right]  (sens) {Sens};
	\draw [very thick, -] (mel) -- ++(-1.85cm,0)
	node[left]  (paris) {Paris};
	\draw [very thick, -] (paris) -- ++(1.1,1.1cm)
	node[right]  (me) {Meaux};
	\path (paris) -- (mel) node [midway,auto=right] {\small $\theta_1$};
	\path (paris) -- (me) node [midway,auto=left] {\small $\theta_{2\,\,\,}$ };
	\path (nem) -- (mel) node [midway,auto=left] {\small $\theta_{3}$};
	\path (sens) -- (mel) node [midway,auto=right] {\small $\theta_{4\,}$};
	\end{tikzpicture}
	\caption{Two different versions of the graph of the Seine network in Figure~\ref{fig:seine_scheme} if nodes with latent variables are suppressed and new edges are drawn between the remaining nodes of the affected parts of the tree.}
	\label{fig:seine_scheme_shrink}
\end{figure}

In Figure~\ref{fig:ABcomparison}, we compare the extremal coefficients induced by model~B to the empirical ones and to those induced by model~A. The extremal coefficients resulting from both models turn out to be rather close, the little black circles lying almost on the diagonal in all four plots. The reason may be that the original tree is small and that not many nodes have been suppressed, whereas in case of many latent variables, the impact of suppressing them may be large. Furthermore, the results may depend on the particular choice of the reduced tree: out of many possibilities two of which shown in Figure~\ref{fig:seine_scheme_shrink} we selected the second one. A final shortcoming of the method of suppressing nodes is that tail dependence cannot be calculated for random vectors involving latent variables.  

\begin{figure}[]
	\includegraphics[scale=0.73]{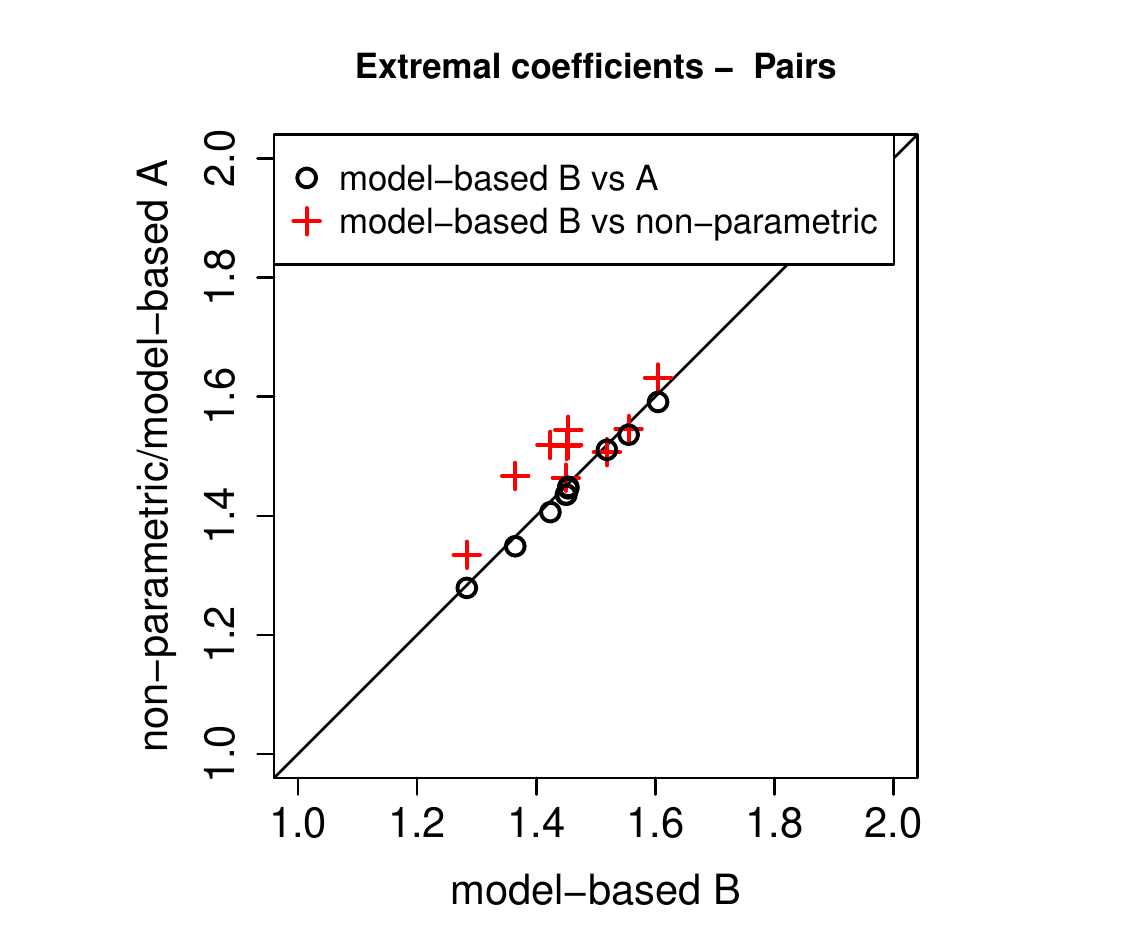}
	\includegraphics[scale=0.73]{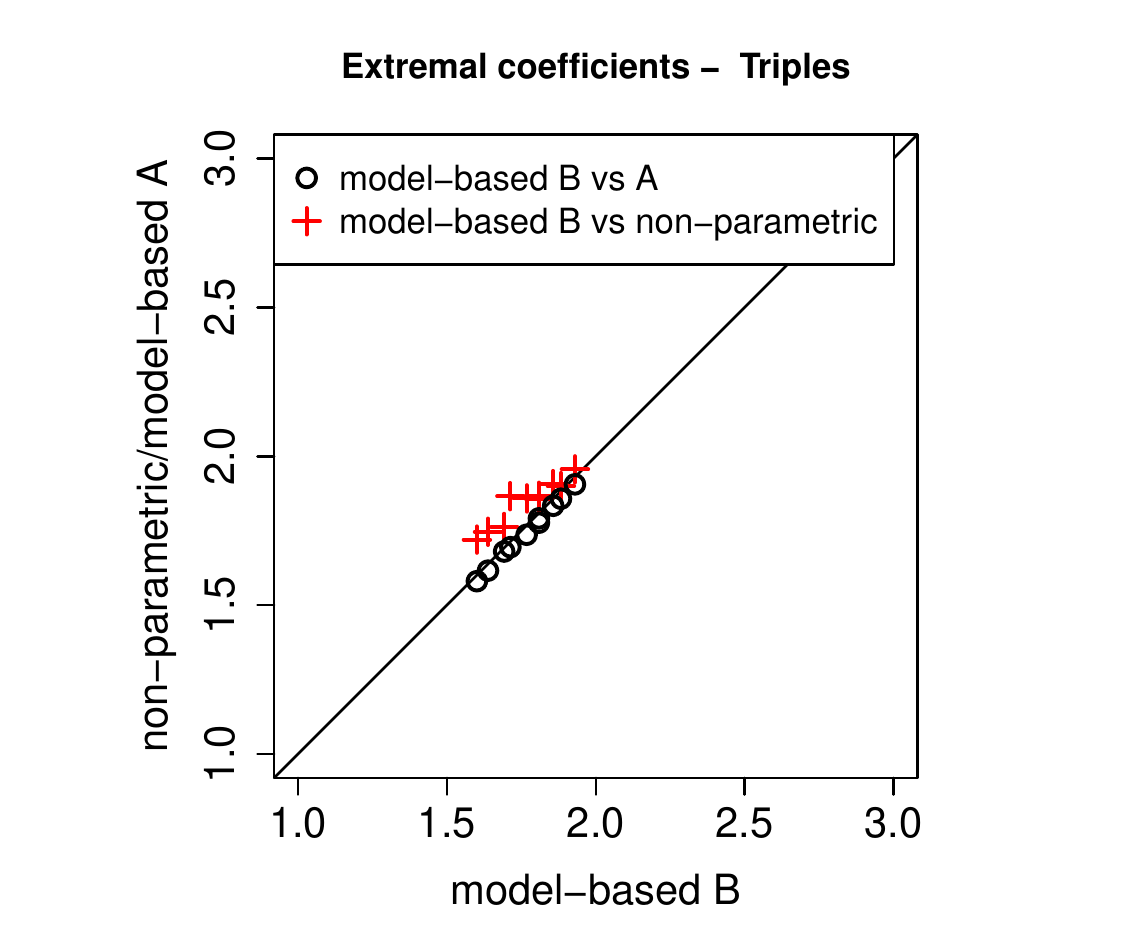}
	\includegraphics[scale=0.73]{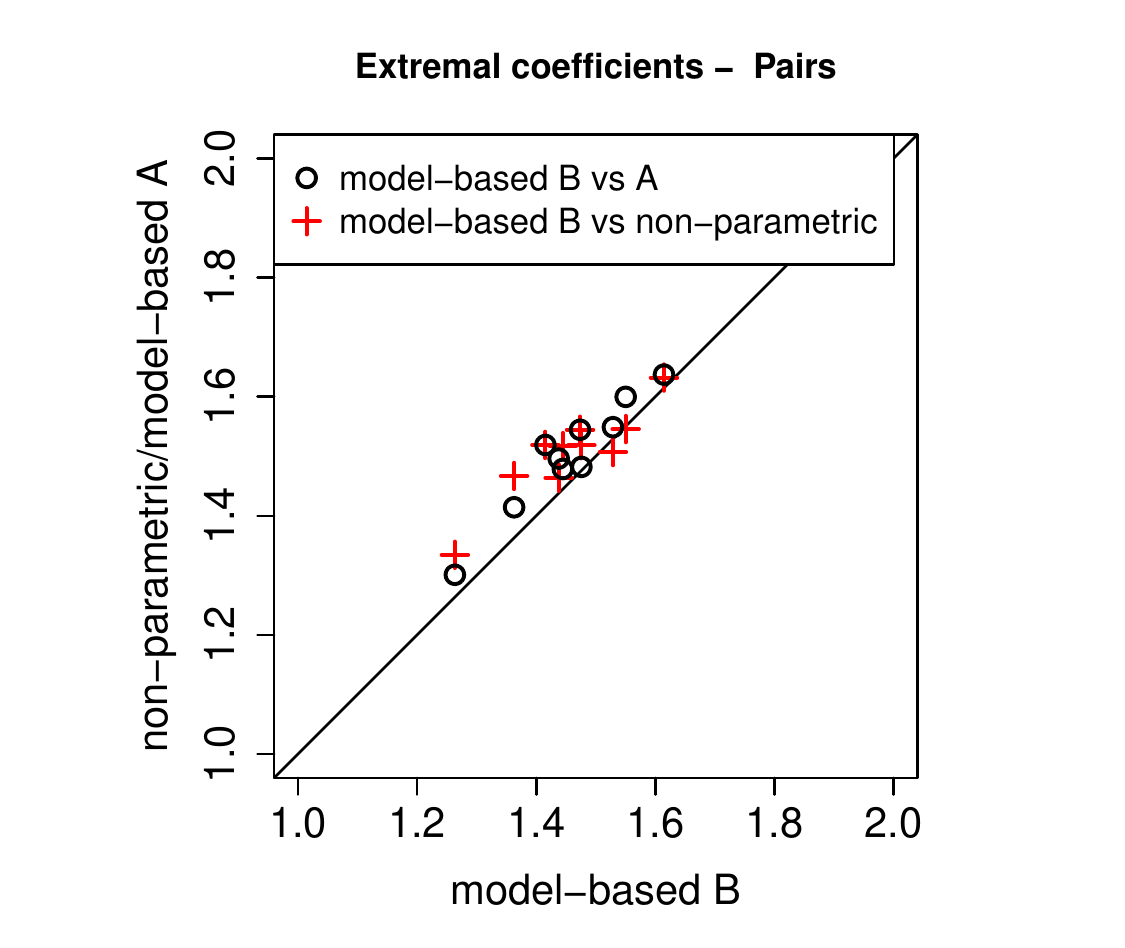}
	\includegraphics[scale=0.73]{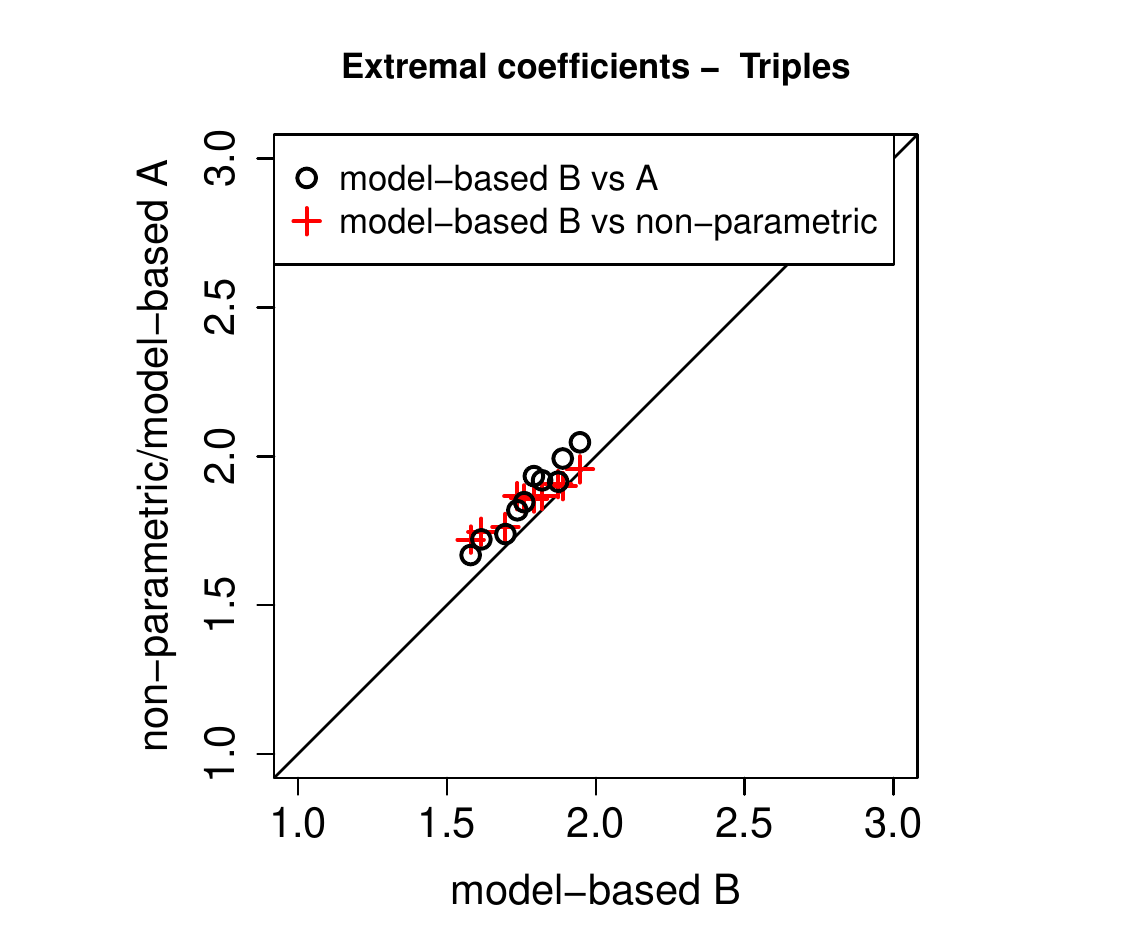}
	\caption{Comparison of extremal coefficients under model~A (latent variables included) and model~B (latent variables excluded). Top: combined MM and CL estimates; bottom: EC estimates. Left: pairs; right: triples.}
	\label{fig:ABcomparison}
\end{figure}

We conclude in Figure~\ref{fig:seine_inv_map_tdf} with a depiction of pairwise upper tail dependence in the Seine network. Shown are the complete and reduced trees with edges weighted by the tail dependence coefficients, defined for a pair $J = \{u, v\} \subset V$ by
\begin{equation}
\label{eq:tdc}
2 \bigl( 1 - \ell_J(1, 1; \theta) \bigr)
\end{equation}
in terms of the pairwise extremal coefficient in \eqref{eq:EC2}. In both trees, the strongest tail dependence occurs along the path from Sens to Paris.  

\begin{figure}
	\centering
	\begin{tabular}{@{}c@{}c}
	\includegraphics[width=0.55\textwidth]{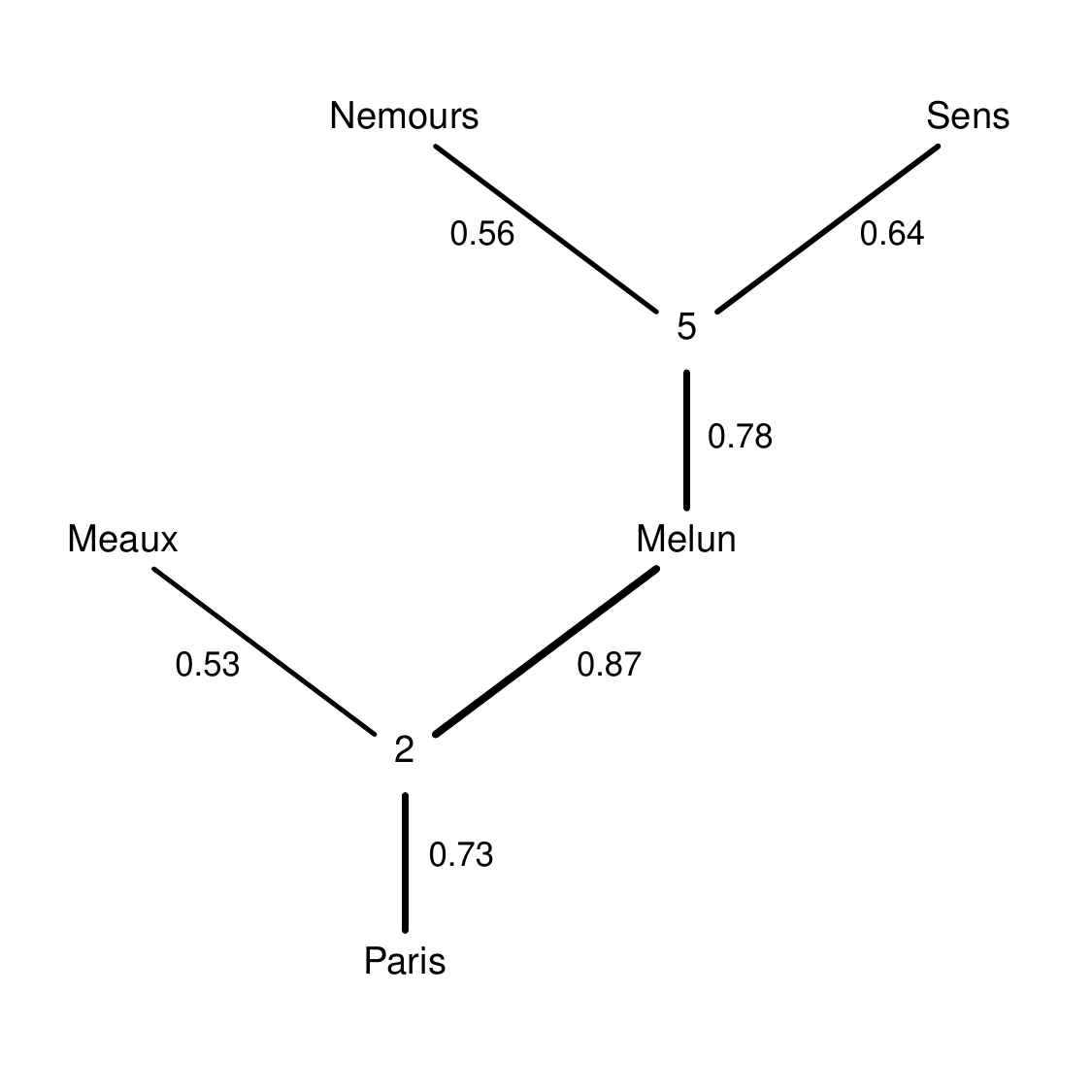}&
	\raisebox{.25\height}{\includegraphics[width=0.42\textwidth]{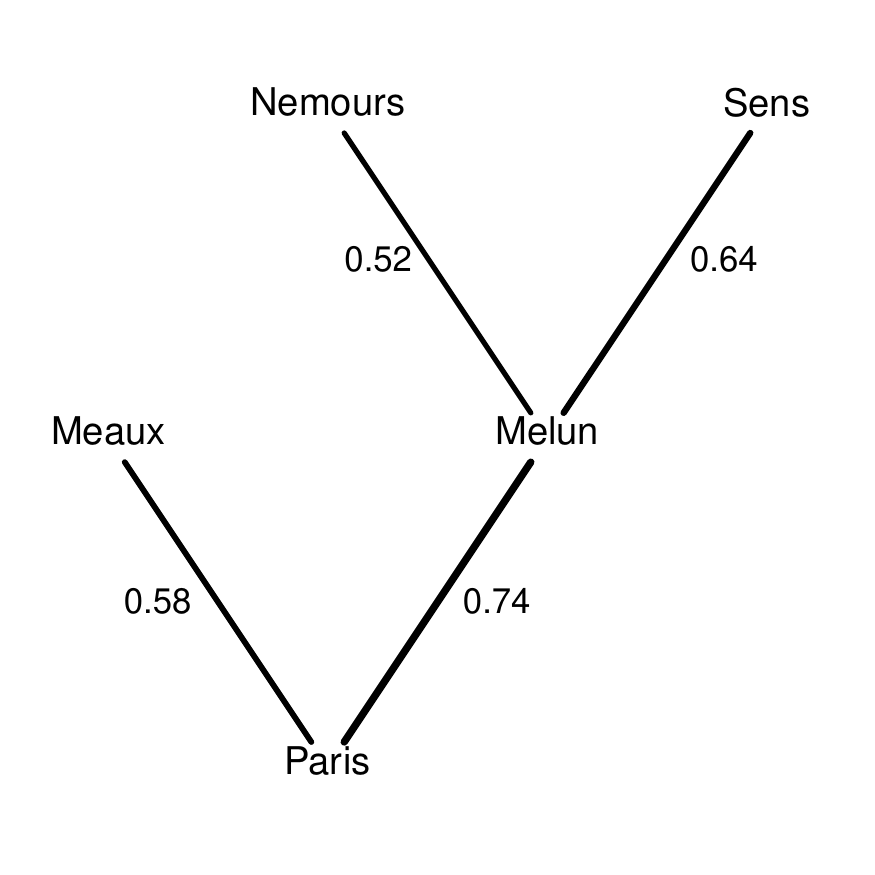}}
	\end{tabular}
	\caption{The trees of the tail dependence coefficients in \eqref{eq:tdc}, using the EC estimates for $\theta$. Left: tree with nodes with latent variables; Right: reduced tree without latent variables.}
	\label{fig:seine_inv_map_tdf}
\end{figure}

\section{Conclusion}
\label{sec:concl}

We have presented a statistical model suitable for studying extremal dependence within a vector of random variables indexed by the nodes of a tree. The edges between the nodes are meant to indicate links between variables arising from a physical or conceptual network, although we do not impose any conditional independence relations. The main assumption is that, upon marginal standardization, the data-generating distribution is in the max-domain of attraction of a max-stable Hüsler--Reiss distribution whose parameter matrix possesses a certain structure induced by the tree: a free parameter is associated to each edge and each element of the Hüsler--Reiss parameter matrix only depends on the sum of the edge parameters along the path between the two corresponding nodes on the tree. We showed that the max-domain of attraction of this tree-structured Hüsler--Reiss distribution contains a specific distribution that, unlike the max-stable Hüsler--Reiss distribution or the associated multivariate Pareto distribution, satisfies the global Markov property with respect to the tree. This auxiliary model not only motivates the postulated structure, it also allowed us to find extremal dependence properties of any distribution satisfying our main assumption.     


The central point and contribution of the paper is related to the identifiability of the edge parameters in case some of the variables are latent (unobservable). This situation occurs for instance in applications on river networks, when measurements on certain locations are missing. We showed that the edge parameters are uniquely identified by the distribution of the observable variables if and only if all nodes indexing latent variables are of degree at least three. Thanks to this result it is possible to quantify tail dependence even between latent variables. The characterization is due to the special structure of the variogram matrix of the Hüsler--Reiss distribution and may not be applicable to other max-stable distributions. 

We fitted the model to water level data on the Seine network on a tree with seven variables, two of which were latent. As the corresponding nodes both had degree three, the six edge parameters were still identifiable and could be estimated based on data from the five observable variables. Three different estimators were proposed and implemented, based on the method of moments, on composite likelihood, and on pairwise extremal coefficients. Comparisons of non-parametric and model-based tail dependence quantities confirmed the adequacy of the fitted structured Hüsler--Reiss distribution.

For comparison we estimated a model where the two nodes with latent variables were suppressed and the edges between the affected parts of the network were redrawn in an arbitrary way. Although for the Seine data this reduction did not have a big impact on the fitted tail dependence model of the observable variables, we argued why it is still recommendable to take latent variables into account, provided there is now a sound way to do so.

An open question concerns parameter identifiability criteria in case of latent variables for Hüsler--Reiss distributions with parameter matrices structured in different ways than in this paper. Even the structure itself may be partially unknown. Another interesting direction for further research concerns extensions from the Hüsler--Reiss family to other parametric families of max-stable distributions.  



\appendix

\section{Appendix}
\subsection{Relation to extremal graphical models in \citet{engelke+h:2020} } \label{app:eh}

\citet{engelke+h:2020} introduce graphical models for extremes in terms of the multivariate Pareto distribution associated to a simple max-stable distribution $G$. We briefly review their approach and compare it with ours in case $G$ is a Hüsler--Reiss distribution. Let $V = \{1,\ldots,d\}$ and let $X = (X_v, v \in V)$ be a random vector with unit-Pareto margins. The condition that $X \in D(G)$ is equivalent to
\[
	\lim_{t \to \infty} \bigl(\P(X_v\leq t z_v, v \in V)\bigr)^t = G(z), \qquad z \in (0, \infty)^V.
\]
By a direct calculation, it follows that
\begin{equation}
\label{eqn:pareto_lim}
	\lim_{t\to\infty} \P \Bigl( X_v/t\leq z_v, v\in V \,\Big\vert\, \max_{v\in V}X_v>t\Bigr)
	=
	\frac{\ln G\bigl(\min(z_v, 1), v\in V\bigr)-\ln G(z)}{\ln G(1, \ldots, 1)},
\end{equation}
for $z \in (0, \infty)^V$, from which
\begin{equation}
\label{eq:Xvtmax2Y}
	(X_v/t, v \in V) \mid \max_{v \in V} X_v > t
	\dto
	Y, \qquad t \to \infty
\end{equation}
where $Y = (Y_v, v \in V)$ is a random vector whose distribution function is equal to the right-hand side in \eqref{eqn:pareto_lim}. The law of $Y$ is a multivariate Pareto distribution, which, upon a change in location, is a special case of the multivariate generalized Pareto distributions arising in \citet{rootzen2006} and \citet[Section~8.3]{beirlant2004statistics} as limit distributions of multivariate peaks over thresholds.

Assuming that $Y$ is absolutely continuous, its support is equal to the L-shaped set $\{ y \in (0, \infty)^V : \max_{v \in V} Y_v > 1 \}$ or a subset thereof, making conditional independence notions related to density factorizations ill-suited for $Y$. This is why \citet{engelke+h:2020} study conditional independence relations for the random vector $Y^u$ defined in distribution as $Y \mid Y_u > 1$ for $u \in V$. 
According to \citet[Definition~2]{engelke+h:2020}, the law of $Y$ is defined to be an extremal graphical model with respect to some graph $\G$ if for all $u\in V$, the law of $Y^u$ satisfies the global Markov property with respect to $\G$. Note that $Y$ itself is not required to satisfy the said Markov property.

The multivariate Pareto distribution derived through \eqref{eqn:pareto_lim} from the Hüsler--Reiss distribution $G=H_{\Lambda}$ is referred to in \citet{engelke+h:2020} as the Hüsler--Reiss Pareto distribution. In their article, the term Hüsler--Reiss graphical models is then used for Hüsler--Reiss Pareto distributions that are extremal graphical models.

To show the relation with our approach, note that \eqref{eq:Xvtmax2Y} implies that, for all $u \in V$, we have
\[
	(X_v/t, v \in V) \mid X_u > t \dto Y^u, \qquad t \to \infty.
\]
Recall the tail tree $(\Xi_{u,v}, v \in V \setminus u)$ in \eqref{eqn:YtoXi} and put $\Xi_{u,u} = 0$. Equations~\eqref{eqn:limitY} and~\eqref{eq:XlimN} in combination with Theorem~2 in \citet{mazo} and the continuous mapping theorem imply that
\[
	(X_v/t, v \in V) \mid X_u > t \dto (\zeta \Xi_{u,v}, v \in V), \qquad t \to \infty,
\]
where $\zeta$ is a unit-Pareto random variable, independent of the log-normal random vector $(\Xi_{u,v}, v \in V)$. Comparing the two previous limit relations, we find that $Y^u$ is equal in distribution to $(\zeta \Xi_{u,v}, v \in V)$. The representation $\ln \Xi_{u,v} = \sum_{e \in \path{u}{v}} \ln M_e$ as path sums starting from $u$ over independent Gaussian increments $\ln M_e$ along the edges implies that the Gaussian vector $(\ln \Xi_{u,v}, v \in V)$ satisfies the global Markov property with respect to $\T$. Since $\zeta$ is independent of $(\Xi_{u,v}, v \in U)$ this Markov property then also holds for $(\zeta \Xi_{u,v}, v \in V)$ and thus also for $Y^u$. But this means exactly that the multivariate Pareto distribution associated to the Hüsler--Reiss distribution with parameter matrix $\Lambda(\theta)$ in \eqref{eqn:lambda} is an extremal graphical model with respect to $\T$.

By way of comparison, the random vector $Z^*$ constructed via properties (Z1)--(Z2) in Section~\ref{ssec:Y} is not max-stable nor multivariate Pareto, but it satisfies the global Markov property with respect to $\T$ and it belongs to $D(H_{\Lambda(\theta)})$, motivating the chosen structure of $\Lambda(\theta)$ in \eqref{eqn:lambda}. In Section~\ref{ssec:introX}, our assumption on $\xi$ after transformation to $X$ with unit-Pareto margins is that $X \in D(H_{\Lambda(\theta)})$. In this sense, we require that the extremal dependence of $\xi$ is like the one of the graphical model $Z^*$. Our approach is thus different from the one in \citet{engelke+h:2020}, who postulate a new definition of extremal graphical models for multivariate Pareto vectors, but without regard for the max-domain of attraction of the corresponding max-stable distributions. Still, for graphical models with respect to trees, both methods arrive at the same structure for the Hüsler--Reiss parameter matrix $\Lambda(\theta)$.

\subsection{Comparison with the Lee--Joe structured Hüsler--Reiss model}


\label{app:LeeJoe}

\citet{joe} already proposed a way to bring structure to the parameter matrix $\Lambda = (\lambda_{ij}^2)_{i,j=1}^d$ of a $d$-variate max-stable Hüsler--Reiss distribution. Recall that \citet{hr} studied the asymptotic distribution of the component-wise maxima of a triangular array of row-wise independent and identically distributed Gaussian random vectors, the $n$-th row having correlation matrix $\rho(n)$. Assuming $(1 - \rho_{ij}(n)) \ln(n) \to \lambda_{ij}^2$ as $n \to \infty$, they found the limit to be the distribution bearing their name. Motivated by this property, \citet{joe} propose to set $\lambda_{ij}^2 = (1 - \rho_{ij}) \nu$ where $\rho = (\rho_{ij})_{i,j=1}^d$ is a structured correlation matrix and $\nu > 0$ is a free parameter. They then introduce the Hüsler--Reiss distributions that result from imposing on $\rho$ the structure of a factor model or the one of a $p$-truncated vine. If $p = 1$, the latter becomes a Markov tree and we can compare their model with ours. In their case, a free correlation parameter $\alpha_e \in (-1, 1)$ is associated to each edge $e \in E$ of the tree on $V = \{1,\ldots,d\}$. The correlation matrix $\rho$ of the resulting Gaussian graphical model is
\[
	\rho_{ij} = \prod_{e \in \path{i}{j}} \alpha_e,
	\qquad i,j \in V.
\]
The Lee--Joe model for the structured Hüsler--Reiss matrix $\Lambda_{\mathrm{LJ}}$ derived from $\rho$ is therefore
\begin{equation}
\label{eqn:LJ}
	\lambda_{ij}^2 
	= (1 - \rho_{ij}) \nu 
	= \left( 1 - \prod_{e \in \path{i}{j}} \alpha_e \right) \nu,
	\qquad i,j \in V.
\end{equation}

The model in \eqref{eqn:LJ} is to be compared with the one in our Eq.~\eqref{eqn:lambda}. The former has $(d-1)+1=d$ free parameters, $(\alpha_e, e \in E)$ and $\nu$, whereas the latter has only $d-1$ free parameters $(\theta_e, e \in E)$. In Eq.~\eqref{eqn:lambda}, the Hüsler--Reiss parameters satisfy
\[
	\lambda_{ij}^2 = \sum_{e \in \path{i}{j}} \lambda_{e}^2,
	\qquad i,j \in V,
\]
where we write $\lambda_e = \lambda_{ab}$ for $e = (a, b) \in E$. In contrast, the Lee--Joe parameter matrix in Eq.~\eqref{eqn:LJ} only satisfies this additivity relation asymptotically as $\nu \to \infty$. For instance, on a tree with $d = 3$ nodes and edges $(1, 2)$ and $(2, 3)$, i.e., a chain, their and our models satisfy respectively
\begin{align*}
	\lambda_{13}^2 
	&= \lambda_{12}^2 + \lambda_{23}^2 - \nu^{-1} \lambda_{12}^2 \lambda_{23}^2
	&& \text{for $\lambda_{ij}^2$ as in Eq.~\eqref{eqn:LJ},} \\
	\lambda_{13}^2 
	&= \lambda_{12}^2 + \lambda_{23}^2 
	&& \text{for $\lambda_{ij}^2$ as in Eq.~\eqref{eqn:lambda}.}
\end{align*}
Since the Lee--Joe parameter $\nu > 0$ takes the role of $\ln(n)$ in the Hüsler--Reiss limit relation, we can think of it as being large. In this interpretation, our parametrization becomes a limiting case of the one of \citet{joe}.

Whereas the Lee--Joe parametrization is motivated from the limit result in \citet{hr} for row-wise maxima of Gaussian triangular arrays, ours is motivated as the max-stable attractor of certain regularly varying Markov trees as in \citet{mazo}, the vector $Z^*$ in Section~\ref{ssec:Y} serving as example. A possible advantage of our structure is that the resulting multivariate Pareto vector falls into the framework of conditional independence for such vectors is an extremal graphical model as in \citet[Definition~2]{engelke+h:2020}, as discussed in Appendix~\ref{app:eh}. In general, this is not true for the multivariate Pareto vector induced by the Lee--Joe structure. For the trivariate tree in the preceding paragraph, for instance, the criterion in Proposition~3 in \citet{engelke+h:2020} is easily checked to be verified for our matrix $\Lambda$ but not for the one of \citet{joe}.

For the Seine data, we compare the fitted Lee--Joe tail dependence model with ours.
In order to avoid possible identifiability issues for the Lee--Joe parameters, we suppress the nodes with latent variables and use the right-hand tree in Figure~\ref{fig:seine_scheme_shrink} for the $d = 5$ observable ones, corresponding to the five locations in the dataset.
The estimation method of \citet{joe} is based on pairwise copulas and annual maxima via composite likelihood.
For year $y$ and for variable $j \in \{1, \ldots, d\}$, let $m_{y,j}$ be the maximum of all observations for that variable and that year, insofar available. These maxima are reported in Table~\ref{annual_maxima} and their availability depends on the variable, i.e., on the location. For Melun there are only 15 such annual maxima in comparison to 33 for Nemours.
For each variable~$j$, transform these maxima to uniform margins $\hat{u}_{y,j}$ using the empirical cumulative distribution function based on all available maxima for that variable. Note that for this transformation, \citet{joe} rely on estimated generalized extreme value distributions instead.
For variables $i, j \in \{1,\ldots,d\}$, let $\mathcal{Y}_{ij}$ be the set of years $y$ for which annual maxima are available for both variables. For the pair (Paris, Meaux) this is the period 1999--2019 while for the pair (Paris, Nemours) this is 1990--2019.
Let $c(u, v; \lambda^2)$ denote the bivariate Hüsler--Reiss copula density with parameter $\lambda^2$.
Following \citet{joe}, we estimate the free parameters in Eq.~\eqref{eqn:LJ} by maximizing a composite likelihood: letting $\lambda_{ij}^2(\alpha, \nu)$ denote the right-hand side in \eqref{eqn:LJ}, the parameter estimates are
\[
	(\hat{\alpha}, \hat{\nu})
	=
	\argmax_{\alpha \in (-1, 1)^{d-1}, \nu \in (0, \infty)}
	\sum_{i, j = 1}^d \sum_{y \in \mathcal{Y}_{ij}}
	\ln c\bigl(\hat{u}_{i,y}, \hat{u}_{j,y}; \lambda_{ij}^2(\alpha, \nu)\bigr).
\]
For the implementation, we relied on the \texttt{R} package \texttt{CopulaModel} \citep{CM}.

Next, we compute bivariate extremal coefficients and compare them with the non-parametric ones on the one hand and with those obtained using our own model on the other hand. The points in Figure~\ref{fig:LJ-ours} being some distance away from the diagonal, the two methods indeed seem to give somewhat different results. Moreover, there is less concordance between the non-parametric estimates and the ones from the Lee--Joe model than between the non-parametric ones and those resulting from our model: compare the red crosses in Figure~\ref{fig:LJ-ours} with those in the left-hand plots in Figure~\ref{fig:ABcomparison}. 
\begin{figure}[h]
	\centering
	\includegraphics[scale=0.75]{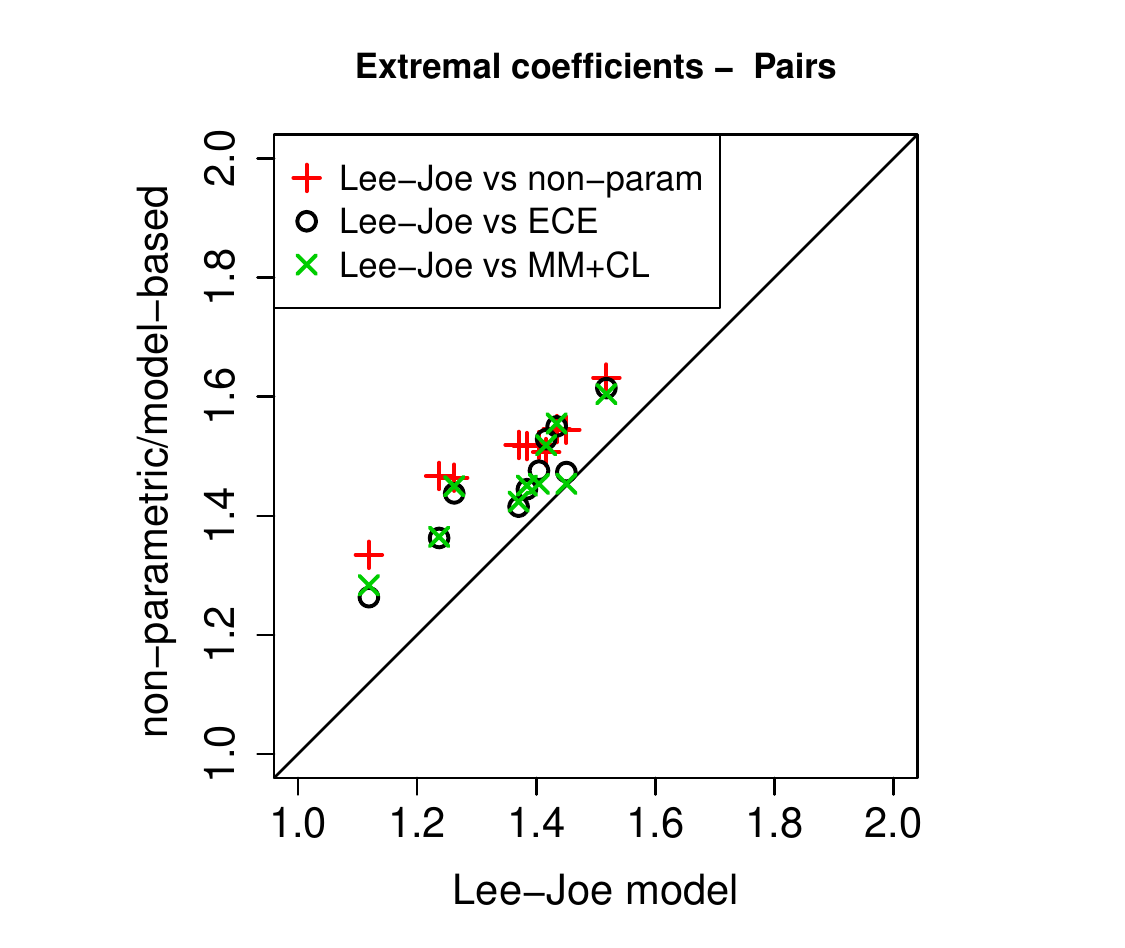}
	\caption{Bivariate extremal coefficients comparison: using the modelling and estimation method of \citet[Section~4--5]{joe} and those proposed in this paper.}
	\label{fig:LJ-ours}
\end{figure}

\subsection{Proof of Proposition~\ref{prop:Ydomain}} 
\label{app:proofprop}

We show that the stdf $l$ of $Z^*$ is equal to $l_U$ in \eqref{eq:stdf_hr} with $U = V$ and $\Lambda = \Lambda(\theta)$. Since the margins of $Z^*$ are unit-Fréchet, they are tail equivalent to the unit-Pareto distribution, so that standardization to the unit-Pareto distribution is unnecessary. By the inclusion--exclusion principle,
\begin{align}
  l(x_1,\ldots,x_d)&
  =
  \lim_{t \to \infty}
  t \, \P(Z_1^*>t/x_1 \text{ or \ldots{} or } Z_d^*>t/x_d)\nonumber \\&
  =
  \sum_{i=1}^d(-1)^{i-1} \sum_{\substack{W\subseteq V\\|W|=i}}
  \lim_{t \to \infty} t \, \P(Z_v^*>t/x_v, v\in W)
\label{eqn:39}
\end{align}
for $x \in (0, \infty)^d$. For any non-empty $W\subseteq V$ and any $u\in W$, it holds by \eqref{eqn:YtoXi} in combination with Theorem~2 in \citet{mazo} that
\begin{align*}
    \lim_{t \to \infty} t \, \P(Z_v^*>t/x_v, v\in W)
	&
    =
    \lim_{t \to \infty} t\frac{1}{t/x_u}
    \P \left( 
    	\frac{Z_u^*}{t/x_u} \frac{Z_v^*}{Z_u^*} > \frac{x_u}{x_v}, \, v\in W \setminus u 
    	\, \Big| \, 
    	Z_u^*>\frac{t}{x_u} 
    \right)\\&
    =
    x_u \, \P(\zeta \Xi_{uv} > x_u/x_v, v\in W \setminus u),
\end{align*}
with $\zeta$ a unit-Pareto variable independent of $(\Xi_{uv}, v \in V \setminus u)$. Using the fact that $1/\zeta$ is a uniform variable on $[0,1]$ and setting $\Xi_{uu}=1$ we have
\begin{align*}
	\lefteqn{
    x_u \, \P(\zeta \Xi_{uv}>x_u/x_v, v\in W\setminus u)
	} \\
    &=
    x_u \, \P\bigl(1/\zeta < \min\{(x_v/x_u)\Xi_{uv}, v\in W\setminus u \}\bigr)\\&
    = x_u \, \E[\min\{1,(x_v/x_u)\Xi_{uv}, v\in W\setminus u\} ]
    =
    \E[\min\{x_v \Xi_{uv}, v\in W\}]
    \\&=
    \int_0^{x_{u}} 
    	\P \bigl(x_v \Xi_{uv}>y, v\in W\setminus {u}\bigr) \, 
    \mathrm{d} y
    \\&=
    \int_{-\ln x_{u}}^{\infty}
    	\P\bigl(
    		\ln \Xi_{uv} > (-\ln x_v)-z, v\in W\setminus u
    	\bigr) \,
    	\exp(-z) \, 
    \mathrm{d} z
\end{align*}
upon a change of variable $y = \exp(-z)$. Since $(\ln \Xi_{uv}, v \in V \setminus u)$ is multivariate normal with mean vector $\mu_{V,u}(\theta)$ and covariance matrix $\Sigma_{V,u}(\theta)$, we obtain from \eqref{eqn:39} that the stdf of $Z^*$ is equal to $- \ln H_{\Lambda(\theta)}(1/x_1, \ldots, 1/x_d)$, with $H_{\Lambda}$ the cumulative distribution function in Eqs.~(3.5)--(3.6) in \citet{hr}, but with unit-Fréchet margins rather than Gumbel ones. By Remark~2.5 in \citet{nikoloulopoulos+j+l:2009}, this stdf is equal to the one given in \eqref{eq:stdf_hr}, as required.

\subsection{Choice of node neighborhoods and parameter identifiability}
\label{app:Wu}

The MM estimator in \eqref{finalmm} and the CL estimator in Section~\ref{ssec:cle} involve the choice of subsets $W_u \subseteq U$ for $u \in U$. These sets or neighborhoods need to be chosen in such a way that the parameter vector $\theta$ is still identifiable from the collection of covariance matrices $\Sigma_{W_u,u}(\theta)$ for $u \in U$ and thus from the path sums $p_{a,b}$ for $a, b \in W_u$ and $u \in U$. Here we illustrate this issue with an example.

Consider the following structure on five nodes where all variables are observable except for the one on node~$2$:
\begin{center}
	%
	\begin{tikzpicture}
	\node[hollow node] (1) at (0,0) {$\xi_1$};
	\node[hollow node, fill=gray] (2) at (1.65,0)  {\textcolor{white}{$\xi_2$}};
	\node[hollow node] (3) at (3.3,0)  {$\xi_3$};
	\node[hollow node] (4) at (4.95,0)  {$\xi_4$};
	\node[hollow node] (5) at (2.85,1)  {$\xi_5$};
	\path (1) edge (2);
	\path (2) edge (3);
	\path (3) edge (4);
	\path (2) edge (5);
	\path (1) -- (2) node [midway,auto=right] {\small $\theta_{12}$};
	\path (2) -- (3) node [midway,auto=right] {\small $\theta_{23}$};
	\path (3) -- (4) node [midway,auto=right] {\small $\theta_{34}$};
	\path (2) -- (5) node [midway,auto=left] {\small $\theta_{25}$};
	\end{tikzpicture} 	
\end{center}
Clearly, the parameter vector $\theta = (\theta_{12}, \theta_{23}, \theta_{34}, \theta_{25})$ is identifiable from the distribution of the observable variables because the criterion of Proposition~\ref{prop:identif} is satisfied: the only node whose variable is latent has degree three.

First, consider the following subsets $W_u$ for $u\in \{1,3,4,5\}$: 
\[
W_1=\{1,5\}, \quad W_3=\{3,4\}, \quad W_4=\{3,4\}, \quad W_5=\{1,5\}.
\] 
The four $1 \times 1$ covariance matrices $\Sigma_{W_u,u}(\theta)$ that correspond to these subsets are 
\begin{align*}
\Sigma_{W_1,1}(\theta)&=\theta_{12}^2+\theta_{25}^2=p_{15}, 
&\Sigma_{W_4,4}(\theta)&=\theta_{34}^2=p_{34},\\
\Sigma_{W_3,3}(\theta)&=\theta_{34}^2=p_{34}, 
&\Sigma_{W_5,5}(\theta)&=\theta_{12}^2+\theta_{25}^2=p_{15}.
\end{align*}
We are not able to identify the parameter $\theta$ because the set of path sums $\{p_{15}, p_{34}\}$ is too small: we have only two equations andfor four unknowns. 

Second, consider instead the following node sets    
\[
W_1=\{1,5,3\}, \quad W_3=\{1,3,4,5\}, \quad W_4=\{3,4\}, \quad W_5=\{1,5\}.
\]
The four covariance matrices $\Sigma_{W_u,u}(\theta)$ are now
\begin{align*}
\Sigma_{W_1,1}(\theta)&=
\begin{bmatrix}
\theta_{12}^2+\theta_{25}^2&\theta_{12}^2\\
\theta_{12}^2&\theta_{12}^2+\theta_{23}^2
\end{bmatrix}=
\begin{bmatrix} 
p_{15}&p_{12}\\p_{12}&p_{13}
\end{bmatrix},
&\Sigma_{W_4,4}(\theta)&=\theta_{34}^2=p_{34}, \\
\Sigma_{W_3,3}(\theta)&=
\begin{bmatrix}
\theta_{12}^2+\theta_{23}^2&0&\theta_{23}^2\\
0&\theta_{34}^2&0\\
\theta_{23}^2&0&\theta_{23}^2+\theta_{25}^2
\end{bmatrix}=
\begin{bmatrix} 
p_{13}&0&p_{23}\\0&p_{34}&0\\p_{23}&0&p_{35}
\end{bmatrix},
&\Sigma_{W_5,5}(\theta)&=
\theta_{12}^2+\theta_{25}^2=p_{15}.
\end{align*}
Clearly, the four edge parameters are identifiable from these covariance matrices.

\subsection{Finite-sample performance of the estimators}
\label{app:simu}

We assess the performance of the three estimators introduced in Section~\ref{section:estim} by numerical experiments involving Monte Carlo simulations. 

Let $\xi'=(\xi'_v, v\in V)$ be a random vector with continuous joint probability density function and satisfying the global Markov property, \eqref{eq:gmp}, with respect to the graph in Figure~\ref{fig:sim_miss}. Let  $f_u(x_u)$ for any $u\in V$ be the marginal density function of the variable $\xi'_u$ and let $x_j \mapsto f_{j\mid v}(x_j\mid x_v)$ be the conditional density function of $\xi'_j$ given $\xi'_v = x_v$. For any $u\in V$ the joint density function of $\xi'$ is 
\begin{equation} \label{eq:jdf}
f(x)=
f_{u}(x_u)\prod_{(v,j)\in E_u}f_{j\mid v}(x_j\mid x_v),
\end{equation}
with $E_u \subseteq E$ the set of edges directed away from $u$, i.e., $(v,j) \in E_u$ if and only if $v = u$ or $v$ separates $u$ and $j$. The joint density $f$ is determined by $d-1$ bivariate densities $f_{vj}$. It would seem that the joint density $f$ depends on $u$, but this is not so, as can be confirmed by writing out the bivariate conditional densities. We make two parametric choices: the univariate margins $f_u$ are unit Fréchet densities,  $f_j(x_j)=\exp(-1/x_j)/x_j^2$ for $x_j \in (0, \infty)$, and the bivariate margins for each pair of variables on adjacent vertices $j,v$ are Hüsler--Reiss distributions with parameter $\theta_{jv}$. Hence, $\xi'$ corresponds to the vector $Z^*$ in Section~\ref{ssec:Y}. 

\begin{figure}[]
	\centering
	\begin{tikzpicture} 
	\node[hollow node](1){$\xi'_2$}
	child[grow=left]{node [hollow node, fill=gray] (a3) {\textcolor{white}{$\xi'_1$}}
		child {node [hollow node] (a4) {$\xi'_6$}}
		child {node [hollow node] (a5) {$\xi'_7$}}}
	child[grow=right]{node [hollow node, fill=gray] (a) {\textcolor{white}{$\xi'_3$}}
		child {node [hollow node] (a1) {$\xi'_5$}}
		child {node [hollow node] (a2) {$\xi'_4$}}
	};
	
	\path (1) -- (a) node [midway,auto=left] {\small $0.3$};
	\path (1) -- (a3) node [midway,auto=right] {\small $0.1$};
	\path (a) -- (a1) node [midway,auto=left] {\small $0.5$};
	\path (a) -- (a2) node [midway,auto=left] {\small $0.8$};
	\path (a3) -- (a4) node [midway,auto=right] {\small $0.2$};
	\path (a3) -- (a5) node [midway,auto=right] {\small $1.2$};
	\end{tikzpicture}
	\caption{Tree used for the graphical model underlying the data-generating process in the simulation study in Appendix~\ref{app:simu}. The value of the parameters are $\theta_{12}=0.1$, $\theta_{23}=0.3$, $ \theta_{34}=0.8$, $\theta_{35}=0.5$, $\theta_{16}=0.2$ and $\theta_{17}=1.2$. Variables $\xi'_1$ and $\xi'_3$ are latent.}
	\label{fig:sim_miss}
\end{figure}
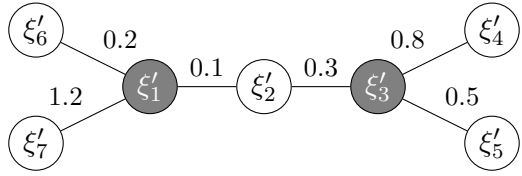

To generate an observation from the left hand-side of \eqref{eq:jdf} above we use the right hand-side of that equation, proceeding iteratively, walking along paths starting from $u$ using the conditional densities.
An observation of $\xi_j'$ given $\xi'_v = x_v$ is generated via the inverse function of the cdf $x_j \mapsto F_{j|v}(x_j \mid x_v)$, the conditional cdf of $\xi'_j$ given $\xi'_v = x_v$. To do so, the equation $F_{j|v}(x_j \mid x_v)-p=0$ is solved numerically as a function in $x_j$ for fixed $p\in(0,1)$. The choice of the Hüsler--Reiss bivariate distribution gives the following expression for $F_{j\mid v}(x_j\mid x_v)$:
\begin{equation*}
\begin{split}
&\Phi\left(
\frac{\theta_{jv}}{2}+\frac{1}{\theta_{jv}}\ln\frac{x_j}{x_v}
\right)
\cdot
\exp\left[
-\frac{1}{x_v}\left\{\Phi\left(\frac{\theta_{jv}}{2}+\frac{1}{\theta_{jv}}\ln\frac{x_j}{x_v}\right)-1\right\}
-
\frac{1}{x_j}\Phi\left(\frac{\theta_{jv}}{2}+\frac{1}{\theta_{jv}}\ln\frac{x_v}{x_j}\right)
\right].
\end{split}
\end{equation*}
After generating all the variables $(\xi'_{v})_{v\in V}$ in this way, independent standard normal noise $\varepsilon\sim \mathcal{N}_{d}(0,I_d)$ is added. Although the distribution of $\xi=\xi'+\varepsilon$ is not necessarily a graphical model with respect to the graph in Figure~\ref{fig:sim_miss}, it is still in the max-domain of attraction of a Hüsler--Reiss distribution with parametric matrix as in \eqref{eqn:lambda}. Hence the vector $\xi$ is still in the class of models under consideration in Section~\ref{ssec:introX}.  
The data on nodes~1 and~3 are discarded and not used in the estimation so as to mimic a model with two latent variables, $\xi_1$ and $\xi_3$; according to Proposition~\ref{prop:identif}, the six dependence parameters are still identifiable. In this way, we generate $200$ samples of size $n = 1000$. The estimators are computed with threshold tuning parameter $k \in \{25,50,100,150,200,300\}$. 

The bias, standard deviation and root mean squared errors of the three estimators are shown in Figure~\ref{fig:theta_miss1-4} and Figure~\ref{fig:theta_miss1-7} for the six parameters. The MME and CLE are computed with the sets $W_u$ being $W_2 = \{2, 4, 5, 6, 7\}$, $W_4 = W_5 = \{2, 4, 5\}$, and $W_6 = W_7 = \{2, 6, 7\}$. As is to be expected, the absolute value of the bias is increasing with $k$, while the standard deviation is decreasing and the mean squared error has a $U$-shape and eventually increases with $k$. The MME and CLE have very similar properties. For larger values of the true parameter, e.g. $\theta_{34}=0.8$ and $\theta_{17}=1.2$, all the three estimators perform in a comparable way. The ECE tends to have larger absolute bias and standard deviation for smaller values of the true parameters.

\begin{figure}[]
	\centering
	\includegraphics[scale=0.82]{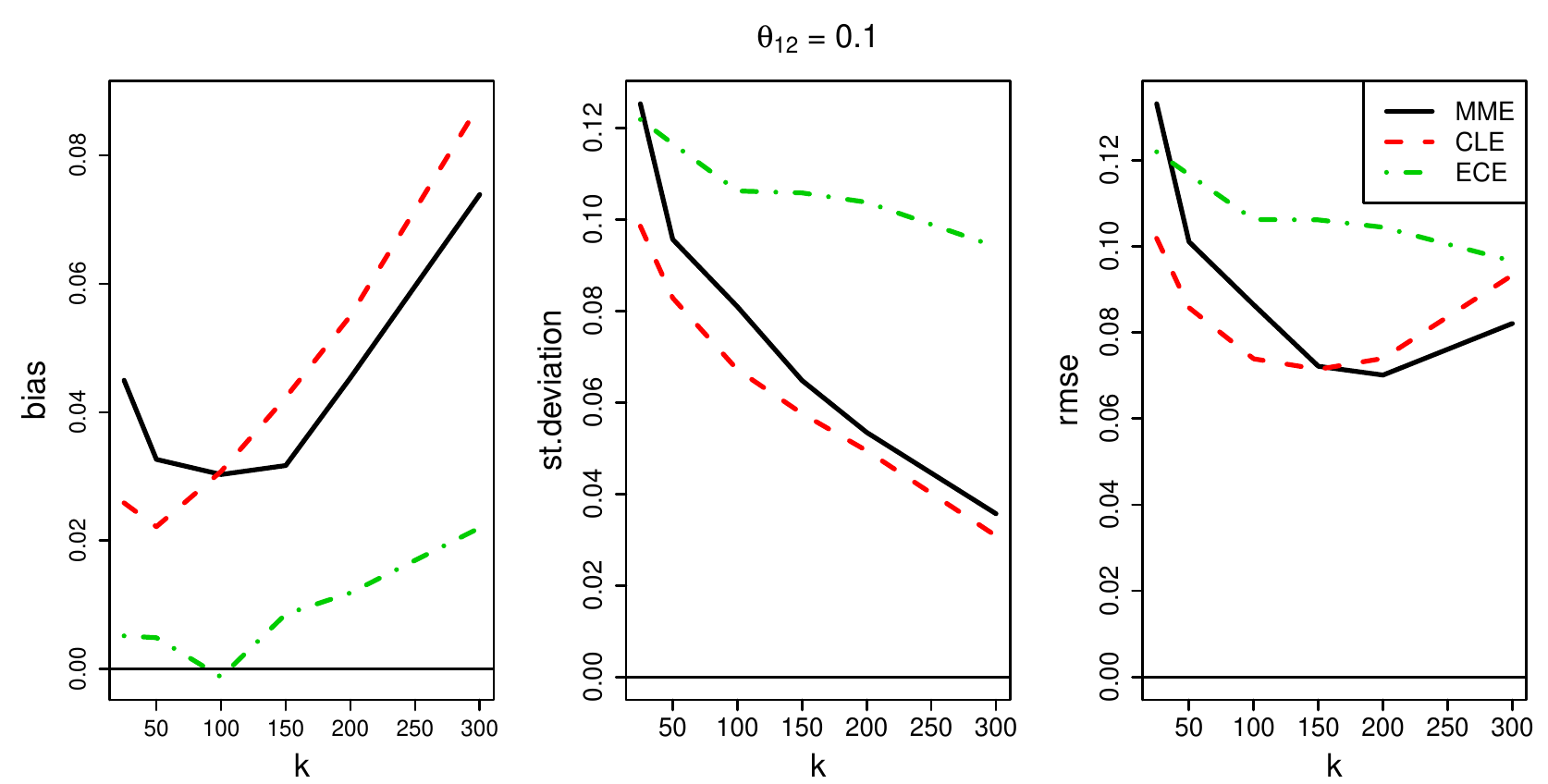}
	\includegraphics[scale=0.82]{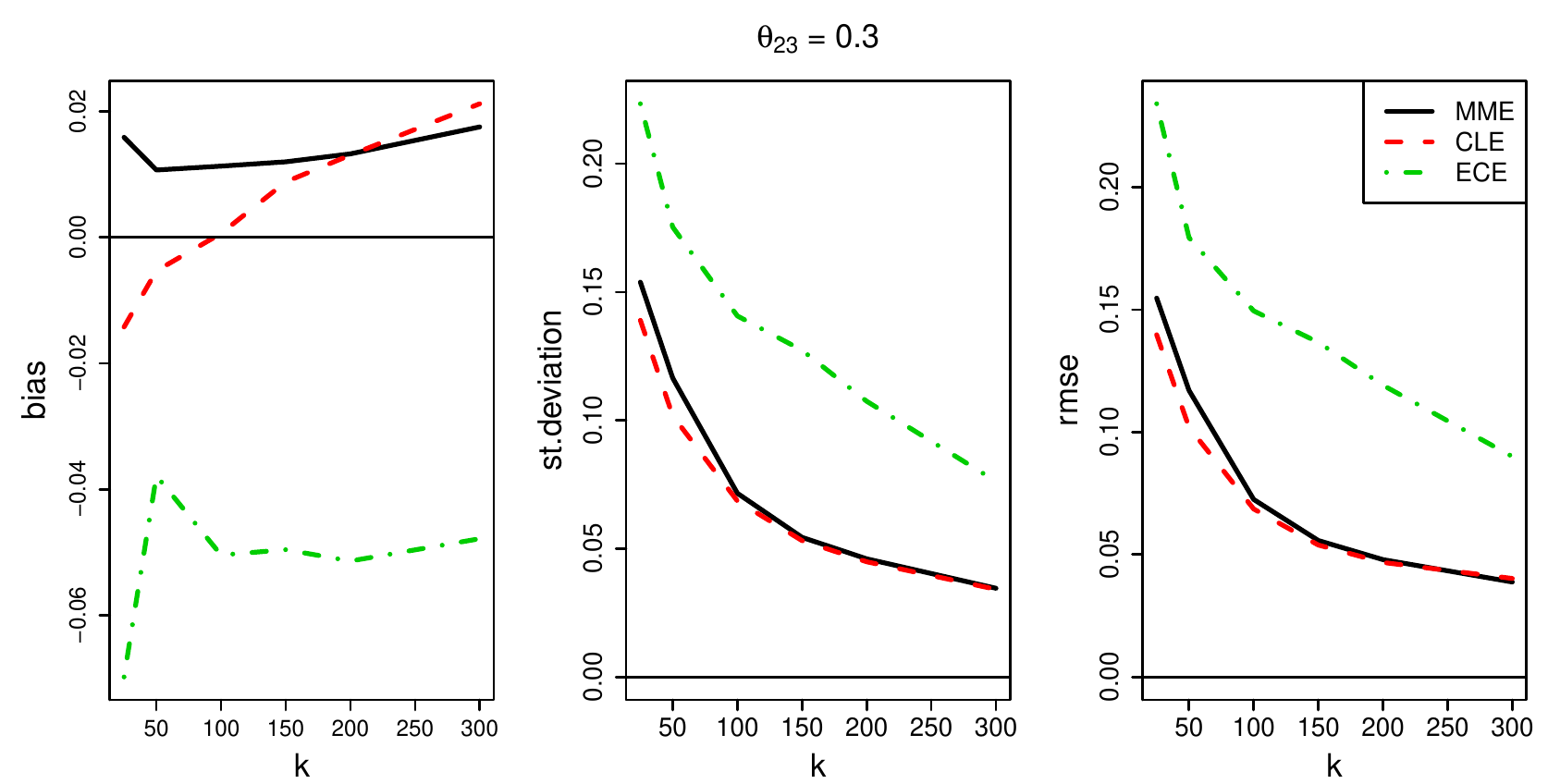}
	\includegraphics[scale=0.82]{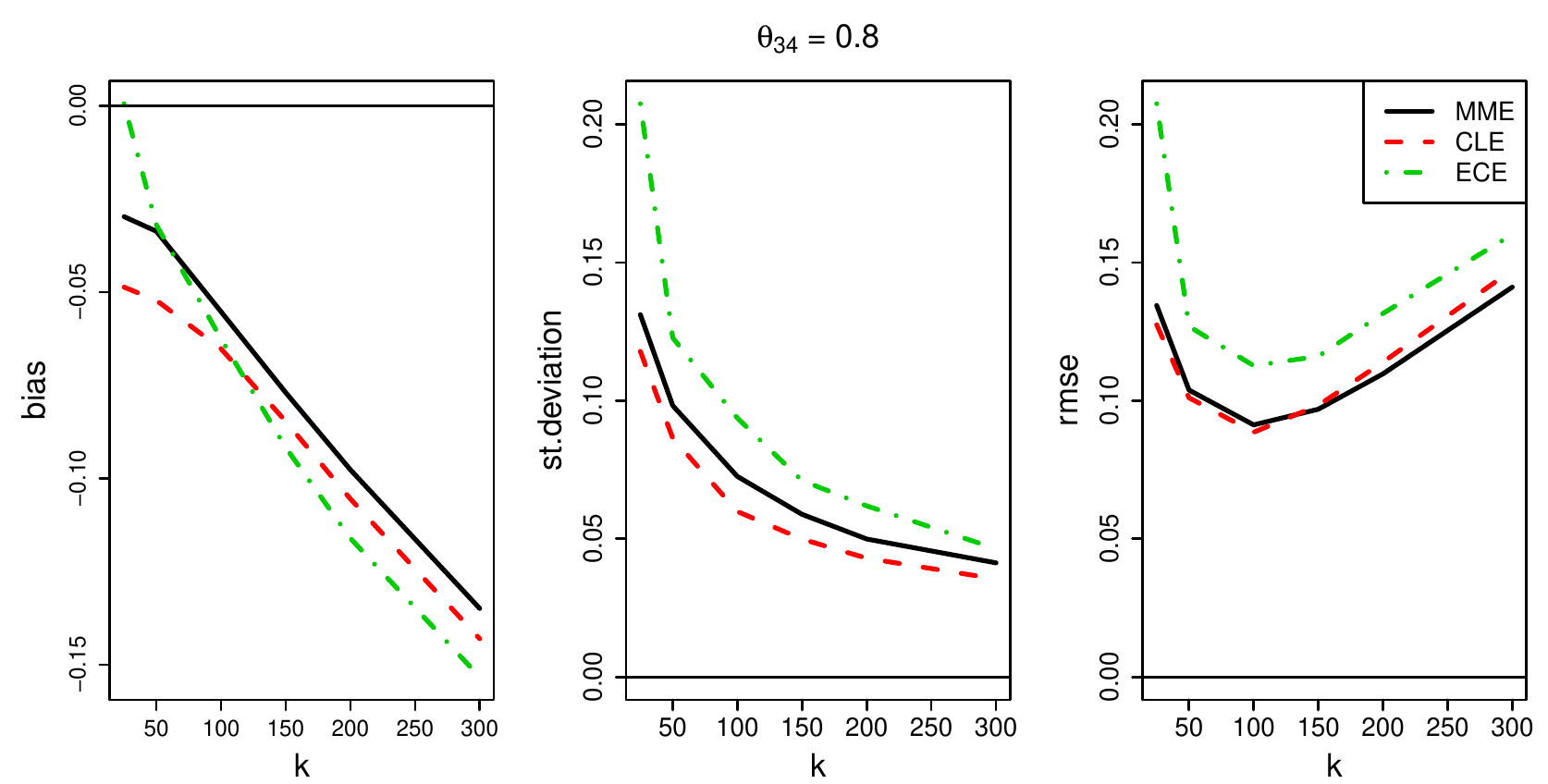}
	\caption{Bias (left), standard deviation (middle) and root mean squared error (right) of the method of moment estimator (MME), composite likelihood estimator (CLE) and pairwise extremal coefficient estimator (ECE) of the parameters $\theta_{12}$ (top), $\theta_{23}$ (middle), and $\theta_{34}$ (bottom) as a function of the threshold parameter $k$. Model and settings as described in Appendix~\ref{app:simu}.}
	\label{fig:theta_miss1-4}
\end{figure}

\begin{figure}[]
	\centering
	\includegraphics[scale=0.82]{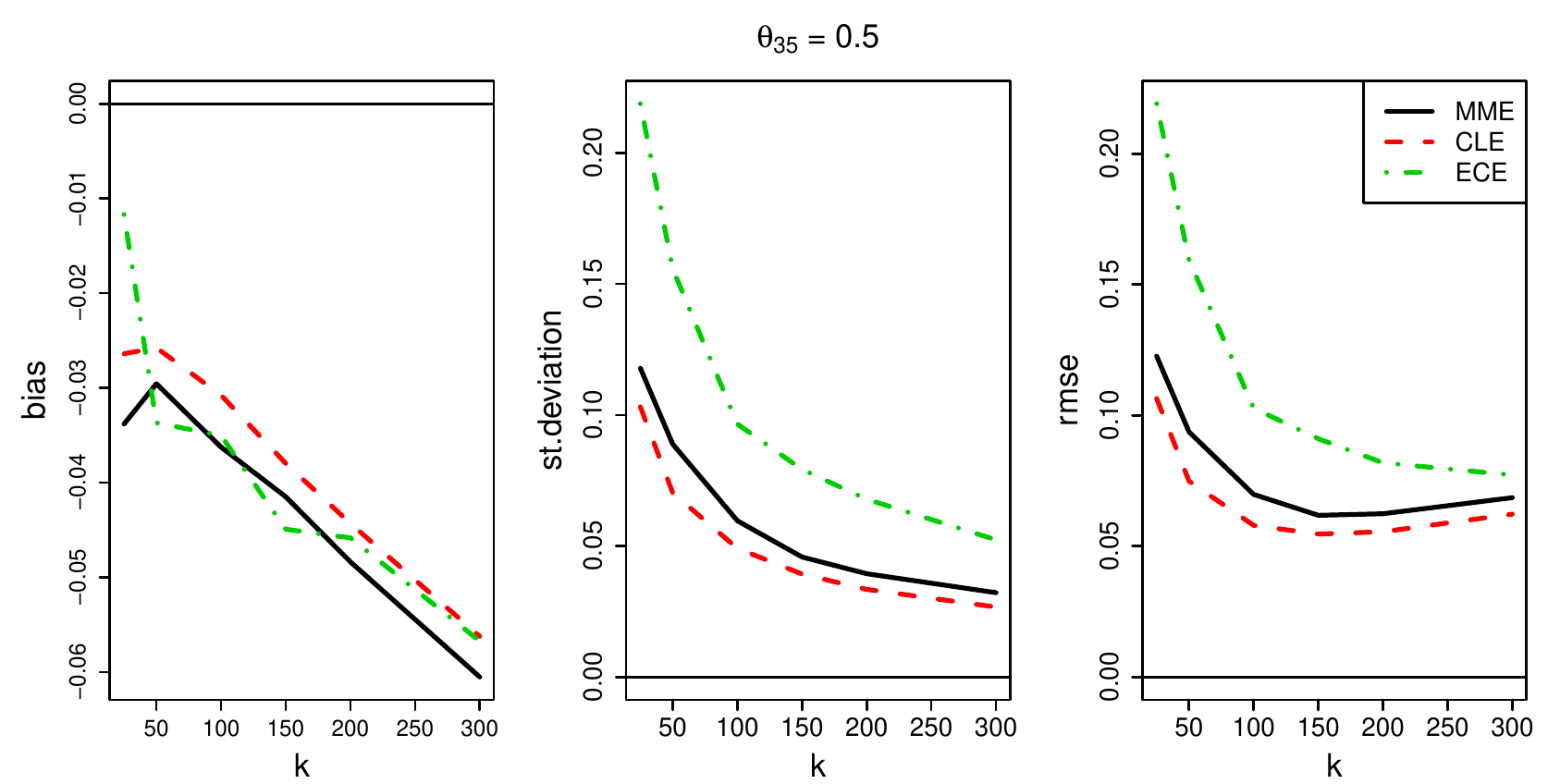}
	\includegraphics[scale=0.82]{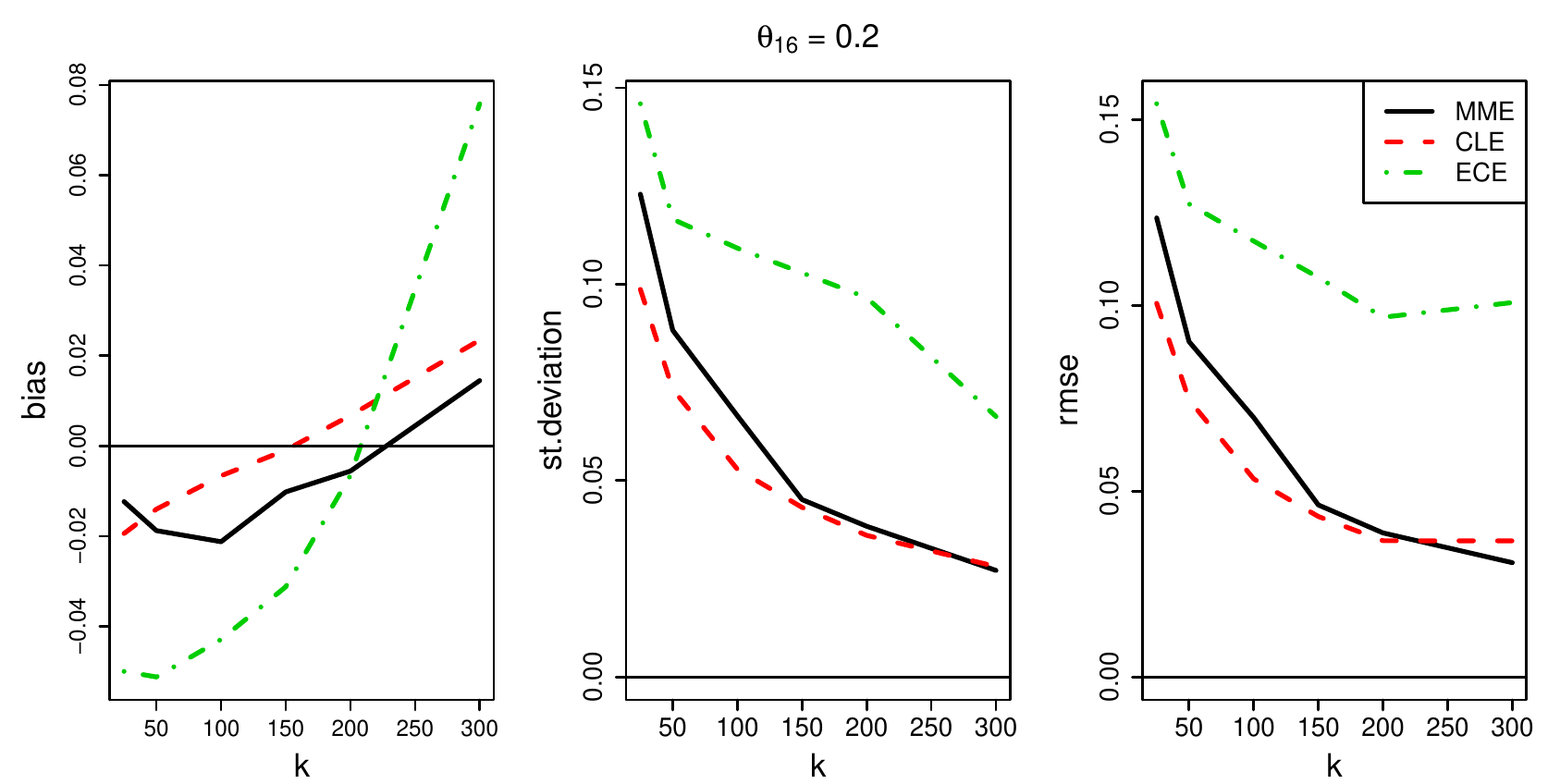}
	\includegraphics[scale=0.82]{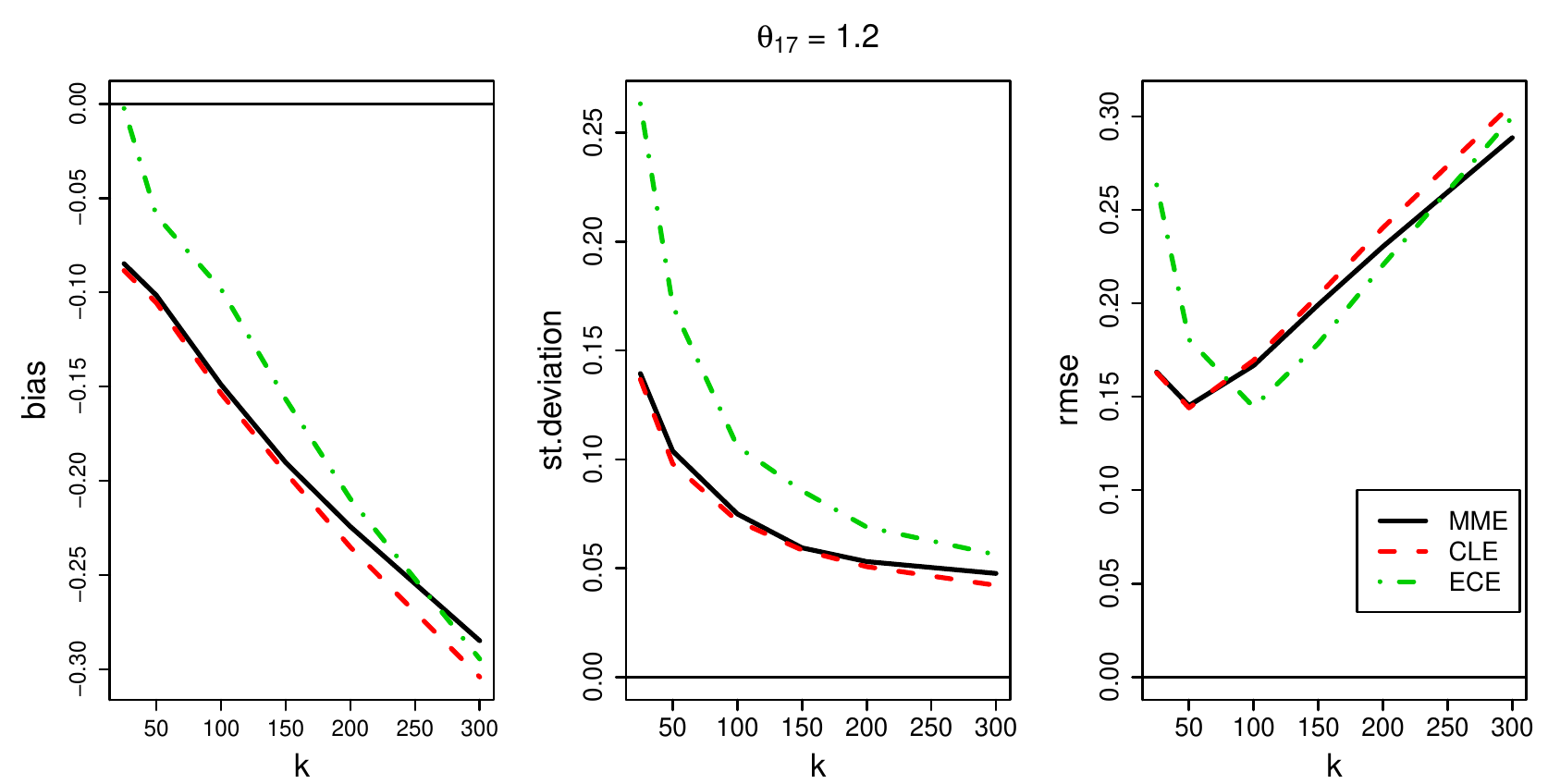}
	\caption{Bias (left), standard deviation (middle) and root mean squared error (right) of the method of moment estimator (MME), composite likelihood estimator (CLE) and pairwise extremal coefficient estimator (ECE) of the parameters $\theta_{35}$ (top), $\theta_{16}$ (middle), and $\theta_{17}$ (bottom) as a function of the threshold parameter $k$. Model and settings as described in Appendix~\ref{app:simu}.}
	\label{fig:theta_miss1-7}
\end{figure}

\subsection{Seine case study: data preprocessing}
\label{app:Seine:preproc}

The data represent water level in centimeters at the five locations mentioned above and were obtained from \emph{Banque Hydro}, http://www.hydro.eaufrance.fr, a web-site of the Ministry of Ecology, Energy and Sustainable Development of France providing data on hydrological indicators across the country. The dataset encompasses the period from January 1987 to April 2019 with gaps for some of the stations. 

Two major floods in Paris make part of our dataset: the one in June 2016 when the water level was measured at \SI{6.01}{\metre} and the one at the end of January 2018 with water levels slightly less than \SI{6}{m} measured in Paris too. A flood of similar magnitude to the ones in 2016 and 2018 occurred in 1982. By way of comparison, the biggest reported\footnote{According to the report of the Organisation for Economic Co-operation and Development (OECD) \emph{Preventing the flooding of the Seine in the Paris – Ile de France region} - p.4. } flood in Paris is the one in 1910 when the level in Paris reached \SI{8.6}{m}. 

Table~\ref{tab:sum_seine} shows the average and the maximum water level per station observed in the complete dataset. The maxima of Paris, Meaux, Melun and Nemours occurred either during the floods in June 2016 or the floods in January 2018, which can be seen from Table~\ref{annual_maxima} which displays the annual maxima at the five locations and the date of occurrence.

\begin{table}[ht]
\centering
\small
\begin{tabular}{rrrrrr}
  \hline
 Station & Paris  & Meaux  & Melun  & Nemours & Sens \\
 Period&1 Jan 1990 -- & 1 Nov 1999 -- & 1 Oct 2005 -- & 16 Jan 1987 -- &1 Jan 1990 --\\
 &9 Apr 2019&9 Apr 2019&9 Apr 2019&9 Apr 2019&9 Apr 2019\\
 (\#obs)&(10,621)&(6,287)&(4,443)&(10,154)&(9,159)\\
  \hline
Mean (cm) & 139.11 & 275.85 & 296.61 & 210.07 & 133.46 \\ 
 Max (cm) & 601.95 & 468.70 & 545.48 & 439.03 & 333.80 \\ 
   \hline
\end{tabular}
\caption{Average and maximum water level per station in the whole dataset.}
    \label{tab:sum_seine}
\end{table}
 
 From Table~\ref{annual_maxima} it can be observed that for many of the years the dates of maxima occurrence identify a period of several consecutive days during which the extreme event took place. For instance the maxima in 2007 occurred all in the period 4--8 March, which suggests that they make part of one extreme event. Similar examples are the periods 25--31 Dec 2010, 4--12 Feb 2013, 2--4 June 2016, etc. For most of the years this period spans between 3 and 7 days. We will take this into account when forming independent events from the dataset. In particular we choose a window of 7 consecutive calendar days within which we believe the extreme event have propagated through the seven locations. We have experimented with different length of that window, namely 3 and 5 days event period, but we have found that the estimation and analysis results are robust to that choice. 

Figure~\ref{fig:am_events} illustrates the water levels attained at the different locations during selected years from Table~\ref{annual_maxima}. The maxima of Sens, Nemours and Meaux seem to be relatively homogeneous compared to the maxima in Paris. 

\begin{table}[]
\centering
	\begin{tabular}{|c|cc|cc|cc|cc|cc|}
		\hline
		Year &
		\multicolumn{2}{c|}{Paris} & 
		\multicolumn{2}{c|}{Meaux}&
		\multicolumn{2}{c|}{Melun}&
		\multicolumn{2}{c|}{Nemours}&
		\multicolumn{2}{c|}{Sens}\\
		\hline
		&date & cm & date & cm & date & cm & date & cm & date & cm \\ 
		\hline
		1987 & n/a & n/a & n/a & n/a & n/a & n/a & 15/11 & 221 & n/a & n/a \\ 
		1988 & n/a & n/a & n/a & n/a & n/a & n/a & 13/02 & 247 & n/a & n/a \\ 
		1989 & n/a & n/a & n/a & n/a & n/a & n/a & 04/03 & 213 & n/a & n/a \\ 
		1990 & 17/02 & 254 & n/a & n/a & n/a & n/a & 03/07 & 217 & 18/02 & 183 \\ 
		1991 & 10/01 & 339 & n/a & n/a & n/a & n/a & 23/04 & 212 & 04/01 & 175 \\ 
		1992 & 06/12 & 293 & n/a & n/a & n/a & n/a & 15/01 & 218 & 06/12 & 170 \\ 
		1993 & 28/12 & 377 & n/a & n/a & n/a & n/a & 26/09 & 217 & 26/12 & 184 \\ 
		1994 & 11/01 & 478 & n/a & n/a & n/a & n/a & 19/10 & 253 & 09/01 & 260 \\ 
		1995 & 30/01 & 500 & n/a & n/a & n/a & n/a & 21/03 & 277 & 28/01 & 259 \\ 
		1996 & 04/12 & 324 & n/a & n/a & n/a & n/a & 03/12 & 219 & 04/12 & 194 \\ 
		1997 & 28/02 & 313 & n/a & n/a & n/a & n/a & 03/07 & 214 & n/a & n/a \\ 
		1998 & 02/05 & 358 & n/a & n/a & n/a & n/a & 21/12 & 216 & n/a & n/a \\
		\rowcolor{bll}
		1999 & 31/12 & 517 & 30/12 & 413 & n/a & n/a & 30/12 & 252 & 31/12 & 259 \\ 
		2000 & 01/01 & 515 & 02/01 & 407 & n/a & n/a & 07/06 & 233 & 01/01 & 239 \\ 
		2001 & 25/03 & 517 & 30/03 & 427 & n/a & n/a & 16/03 & 260 & 17/03 & 334 \\ 
		2002 & 03/03 & 410 & 03/03 & 403 & n/a & n/a & 01/01 & 272 & 01/01 & 200 \\
		\rowcolor{bll}
		2003 & 08/01 & 410 & 09/01 & 331 & n/a & n/a & 05/01 & 253 & 06/01 & 182 \\ 
		\rowcolor{bll}
		2004 & 21/01 & 372 & 21/01 & 383 & n/a & n/a & 16/01 & 230 & 20/01 & 205 \\ 
		2005 & 17/02 & 192 & 22/01 & 296 & 07/12 & 306 & 24/01 & 217 & 16/02 & 152 \\ 
		2006 & 14/03 & 340 & 08/10 & 333 & 13/03 & 357 & 11/03 & 219 & 12/03 & 223 \\ 
		\rowcolor{bll}
		2007 & 05/03 & 308 & 08/03 & 339 & 05/03 & 333 & 04/03 & 217 & 05/03 & 176 \\ 
		2008 & 29/03 & 301 & 01/01 & 250 & 23/03 & 342 & 15/04 & 219 & 23/03 & 167 \\ 
		2009 & 26/01 & 169 & 03/09 & 288 & 25/12 & 311 & 25/01 & 218 & 25/01 & 152 \\ 
		\rowcolor{bll}
		2010 & 28/12 & 387 & 31/12 & 355 & 27/12 & 390 & 25/12 & 230 & 26/12 & 220 \\ 
		2011 & 01/01 & 337 & 07/01 & 347 & 18/12 & 356 & 09/10 & 287 & 18/12 & 167 \\ 
		2012 & 09/01 & 330 & 23/12 & 308 & 09/01 & 353 & 05/01 & 220 & 08/01 & 186 \\ 
		\rowcolor{bll}
		2013 & 09/02 & 390 & 12/02 & 347 & 05/02 & 366 & 04/02 & 252 & 07/05 & 221 \\ 
		2014 & 03/03 & 273 & 13/12 & 295 & 16/02 & 321 & 02/03 & 226 & 15/02 & 157 \\ 
		2015 & 07/05 & 347 & 21/11 & 295 & 07/05 & 389 & 05/05 & 255 & 06/05 & 211 \\ 
		\rowcolor{bll}
		2016 & 03/06 & 602 & 03/06 & 329 & 03/06 & 545 & 02/06 & 439 & 04/06 & 235 \\ 
		2017 & 07/03 & 243 & 28/12 & 304 & 12/01 & 307 & 08/03 & 221 & 08/03 & 151 \\ 
		\rowcolor{bll}
		2018 & 29/01 & 586 & 02/02 & 469 & 28/01 & 488 & 24/01 & 264 & 26/01 & 288 \\ 
		2019 & 03/02 & 222 & 31/03 & 292 & 22/01 & 314 & 02/02 & 216 & 26/02& 149 \\ 
		\hline
	\end{tabular}
\caption{Annual maxima for all stations. We highlighted some of the years where there is a clear indication that the dates of the occurrence of the maxima at the different locations form a period of several consecutive days. The maxima attained during this period across stations can thus be considered as one extreme event. The water level in centimeters is rounded to the nearest integer.}
\label{annual_maxima}
\end{table}

\begin{figure}[ht]
    \centering
    \includegraphics[scale=0.75]{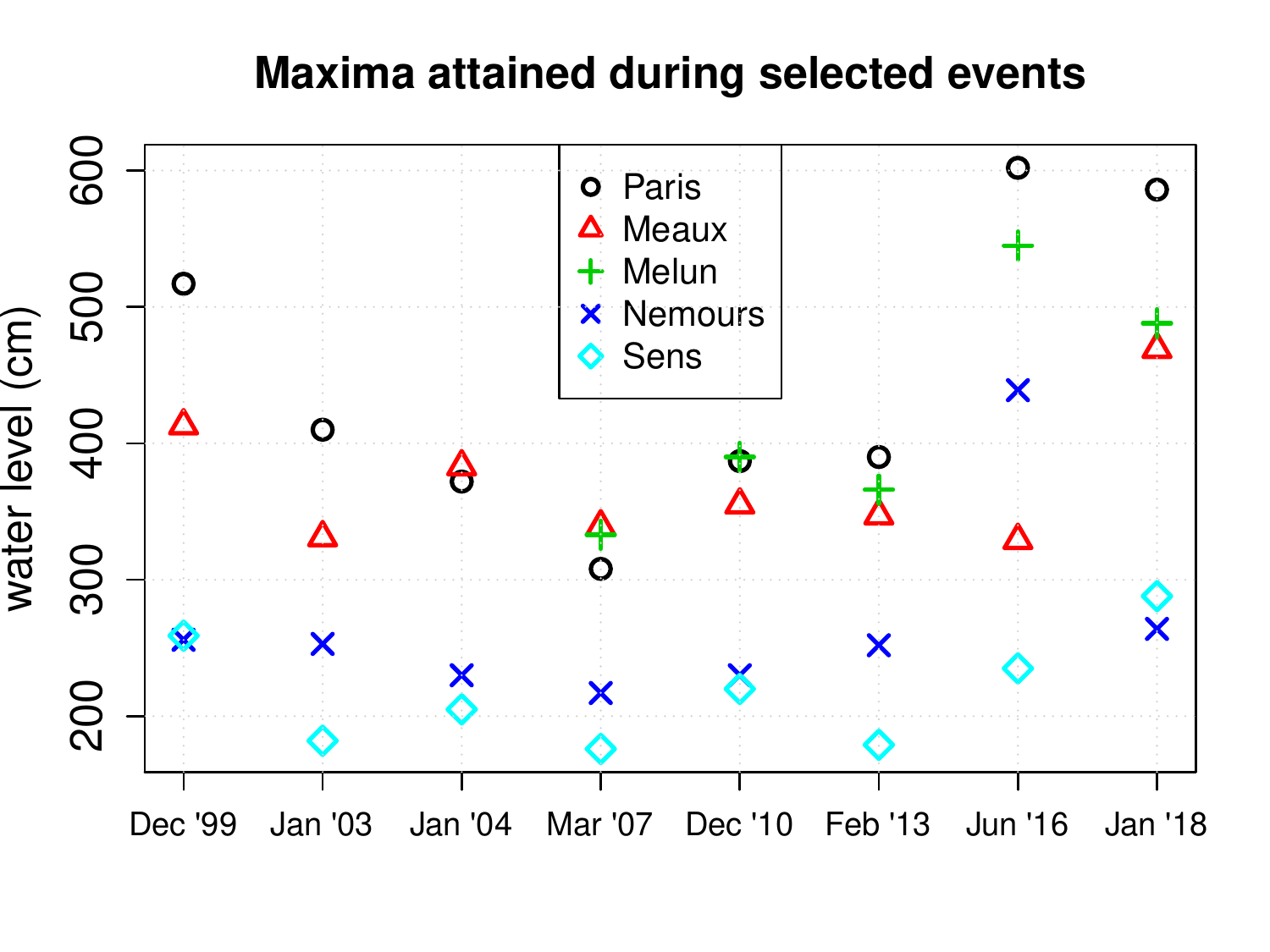}
    \caption{Plot of maxima attained at each location during selected events from Table~\ref{annual_maxima}. }
    \label{fig:am_events}
\end{figure}

For all of the stations water level is recorded several times a day and we take the daily average  to form a dataset of daily observations. Accounting for the gaps in the mentioned period (see Table~\ref{tab:sum_seine} and Table~\ref{annual_maxima}) we end up with a dataset of 3408 daily observations in the period from 1 October 2005 to 8 April 2019. The dataset represents five time series each of length 3408. We consider two sources of non-stationarity: seasonality and serial correlation. 

The serial correlation can be due to closeness in time or presence of long term time trend in the observations. We first apply a declustering procedure, similar to the one in \citet{asadi} in order to form a collection of supposedly independent events. As a first step each of the series is transformed to ranks and the sum of the ranks is computed for every day in the dataset. The day with the maximal rank is chosen, say $d^*$. A period of $2r+1$ consecutive days, centered around $d^*$ is considered and only the observations falling in that period are selected to form the event. Within this period the station-wise maximum is identified and the collection of the station-wise maxima forms one event. Because there is some evidence that the time an extreme event takes to propagate through the seven nodes in our model is about 3--7 days, we choose $r=3$, hence we consider that one event lasts 7 days. In this way we obtain 717 observations of supposedly independent events. As it was mentioned the results are robust to the choice of $r=\{1,2,3\}$.

We test for seasonality and trends each of the series (each having 717 observations). The season factor is significant across all series and the time trend is marginally significant for some of the locations. We used a simple time series model to remove these non-stationarities. The model is based on season indicators and a linear time trend
\begin{equation} \label{eqq:time_series}
    X_{t}=\beta_0+\beta_1\mathbbm{1}_{\text{spring}_t}+\beta_2\mathbbm{1}_{\text{summer}_t}+\beta_3\mathbbm{1}_{\text{winter}_t}+\alpha t+\epsilon_t,
\end{equation}
where $\epsilon_t$ for $t=1,2,\ldots $ is a stationary mean zero process. 
After fitting the model in \eqref{eqq:time_series} to each of the five series through ordinary least squares we obtain the residuals and use those in the estimation of the extremal dependence.

\subsection{ECE-based confidence interval for the dependence parameters}
\label{app:ECE}

Let $\hat{\theta}_{n,k} = \hat{\theta}_{n,k}^{\mathrm{ECE}}$ denote the pairwise extremal coefficient estimator in \eqref{eqn:ECE} and let $\theta_0$ denote the true vector of parameters. By \citet[Theorem~2]{eks16} with $\Omega$ equal to the identity matrix, the ECE is asymptotically normal,
\[
 \sqrt{k}(\hat{\theta}_{n,k}-\theta_0)\dto \mathcal{N}_{|E|}\bigl(0, M(\theta_0)\bigr), \qquad n \to \infty,
\]
provided $k = k_n \to \infty$ such that $k/n \to 0$ fast enough \citep[Theorem~4.6]{einmahl12}. The asymptotic covariance matrix takes the form
\[
    M(\theta_0)
    = (\dot{L}^\top\dot{L})^{-1}
    \dot{L}^\top\Sigma_L\dot{L}
    (\dot{L}^\top\dot{L})^{-1}\, .
\]
The matrices $\dot{L}$ and $\Sigma_L$ depend on $\theta_0$ and are described below. For every $k$ and every $e \in E$, an asymptotic 95\% confidence interval for the edge parameter $\theta_{0,e}$ is given by
\[
	\theta_{0,e} \in \left[\hat{\theta}_{k,n;e}\pm 1.96\sqrt{\{M(\hat{\theta}_{k,n})\}_{ee}/k}\right] .
\]

First, recall that $\mathcal{Q} \subseteq \{ J \subseteq U : |J| = 2 \}$ is the set of pairs on which the ECE is based and put $q = |\mathcal{Q}|$. Define the $\mathbb{R}^q$-valued map $L(\theta) = \bigl(l_J(1,1;\theta), J \in \mathcal{Q}\bigr)$ and let $\dot{L}(\theta) \in \mathbb{R}^{q \times |E|}$ be its matrix of partial derivatives. For a pair $J = \{u, v\}$ and an edge $e=(a,b)$, the partial derivative of $l_J(1,1;\theta)$ with respect to $\theta_{e}$ is given by
\begin{equation*} 
    \frac{\partial  l_J(1,1;\theta)}{\partial \theta_{e}}
    =
    \frac{\phi\left(\sqrt{p_{uv}}/2\right)}
    {\sqrt{p_{uv}}}
    \theta_{e}\mathbbm{1}_{\{e\in\path{u}{v}\}},
\end{equation*}
where $p_{uv}$ is the path sum as in \eqref{eq:sumpath} and $\phi$ denotes the standard normal density function. The partial derivatives of $l_J(1, 1;\theta)$ with respect to $\theta_e$ for every $e\in E$ form a row of the matrix $\dot{L}(\theta)$.

Second, $\Sigma_L(\theta_0)$ is the $q\times q$ covariance matrix of the asymptotic distribution of the empirical stdf,
\[
\big\{\sqrt{k}\big(\hat{l}_{j;n,k}(1,1)
-
l_J(1,1;\theta_0)\big)\big\}_{m=1,\ldots, q}
\dto
\mathcal{N}_{q}(0, \Sigma_L(\theta_0)),
\qquad n \to \infty.
\]
The elements of the matrix $\Sigma_L(\theta_0)$ are defined in terms of the stdf evaluated at different coordinates and of the partial derivatives of the stdf $l(x; \theta)$ with respect to the elements of $x$. For details we refer to \citet[Section~2.5]{eks16}. Here we note that the partial derivatives just mentioned are
\[
	\left. \frac{\partial l_J(x_u,x_v;\theta)}{\partial x_u} \right|_{(x_u,x_v)=(1,1)}
	=
	\Phi(\sqrt{p_{uv}}/2),
	\qquad J = \{u, v\}.
\]

\subsection{Bootstrap confidence interval for the Pickands dependence function}
\label{app:stdf:CI}

For assessing the goodness-of-fit of the proposed model (Section~\ref{ssec:gof}), we construct non-parametric 95\% confidence intervals for $A(w)=l(1-w,w)$ for $w\in [0,1]$. As shown in \citet[Section~5]{kiril_seg_taf} this can be achieved by resampling from the empirical beta copula. For every fixed $w\in [0,1]$ we seek with $a(w)$ and $b(w)$ such that
\[
	\P\bigl(a(w) \leq \hat{l}_{n,k}(1-w,w)-l(1-w,w)\leq b(w)\bigr)
	= 0.95\,,
\]
where $\hat{l}_{n,k}$ is the non-parametric estimator of the stdf. For $a(w)$ and $b(w)$ satisfying the above expression, a point-wise confidence interval is given by 
\begin{equation}
\label{eqn:l:CI}
	A(w) \in \left[\hat{l}_{n,k}(1-w,w)-b(w), \hat{l}_{n,k}(1-w,w)-a(w)\right] .
\end{equation}
Let $(Y^\ast_{v,i})_{v \in U}$, for $i=1,\ldots, n$, be a random sample from the empirical beta copula drawn according to steps A1--A4 of \citet[Section~5]{kiril_seg_taf}. Let the function $\hat{l}^{\beta}_{n,k}$ be the empirical beta stdf based on the original data and let the function $\hat{l}^{\ast}_{n,k}$ be the non-parametric estimate of the stdf using the bootstrap sample.

We use the distribution of $\hat{l}_{n,k}^\ast-\hat{l}^{\beta}_{n,k}$ conditionally on the data as an estimate of the distribution of $\hat{l}_{n,k}-l$. Hence, we estimate $a(w)$ and $b(w)$ by $a^*(w)$ and $b^*(w)$ respectively defined implicitly by
\begin{align*}
	0.95&=
	\P^*\left(
		a^*(w)
		\leq \hat{l}_{n,k}^\ast(1-w,w)-\hat{l}^\beta_{n,k}(1-w,w)
		\leq b^*(w)
	\right)
	\\&=
	\P^*\left(
		a+\hat{l}^\beta_{n,k}(1-w,w)
		\leq \hat{l}_{n,k}^\ast(1-w,w)
		\leq b + \hat{l}^\beta_{n,k}(1-w,w)
	\right) .
\end{align*}
We further estimate the bootstrap distribution of $\hat{l}^\ast_{n,k}$ by a Monte Carlo approximation obtained by $N = 1000$ samples of size $n$ from the empirical beta copula. As a consequence, the lower and upper bounds for $\hat{l}^\ast_{n,k}(1-w,w)$ above are equated to the empirical 0.025- and 0.975-quantiles, respectively, yielding
\begin{align}
\label{eqn:a*b*}
	\hat{l}^\ast_{0.025}(w,1-w)
	&= a^*(w) + \hat{l}^\beta_{n,k}(w,1-w), &
	\hat{l}^\ast_{0.975}(w,1-w)
	&= b^*(w) + \hat{l}^\beta_{n,k}(w,1-w).
\end{align}
Replacing $a(w)$ and $b(w)$ in \eqref{eqn:l:CI} by $a^*(w)$ and $b^*(w)$ respectively as solved from \eqref{eqn:a*b*} yields the bootstrapped confidence interval for $A(w)$ shown in Figure~\ref{fig:pdf}.

\section*{Acknowledgements}
We wish to thank the two Reviewers and the Associate Editor for a wealth of valuable remarks and suggestions, which helped us to enhance the understanding of the place of our models in the existing literature and of our contribution to it. Stefka Asenova is also grateful to David Lee for the clarifications and indications on the model in \citet{joe}.

\small
\bibliography{mypr1_rev2_main}
\bibliographystyle{apalike}

\end{document}